\documentclass[10pt,draftcls,onecolumn]{IEEEtran}
\ifCLASSINFOpdf
\else
   \usepackage[dvips]{graphicx}
\fi
\usepackage{url}
\usepackage{graphicx}
\usepackage{amsmath}
\usepackage{comment}
\usepackage[utf8]{inputenc}
\usepackage[english]{babel}
\usepackage{color}
\usepackage{framed}
\usepackage{amsmath,amssymb,amsfonts}

\usepackage{graphicx}
\usepackage{upgreek}
\usepackage{textcomp}

\usepackage{url}

\usepackage{cite}

\usepackage[dvipsnames]{xcolor}
\usepackage{amsmath}
\usepackage{amsthm}
\usepackage{dsfont}
\usepackage{thmtools}
\usepackage{amssymb}

\def\bGamma{\boldsymbol{\Gamma}}

\def\bSigma{\boldsymbol{\Sigma}}

\def\bOmega{\boldsymbol{\Omega}}

\def\mbb{\mathbf{b}}

\def\mbe{\mathbf{e}}

\def\mbg{\mathbf{g}}

\def\mbn{\mathbf{n}}

\def\mbp{\mathbf{p}}

\def\mbs{\mathbf{s}}
\def\mbt{\mathbf{t}}
\def\mbu{\mathbf{u}}
\def\mbv{\mathbf{v}}

\def\mbx{\mathbf{x}}
\def\mby{\mathbf{y}}
\def\mbz{\mathbf{z}}

\def\mbA{\mathbf{A}}
\def\mbB{\mathbf{B}}
\def\mbC{\mathbf{C}}

\def\mbG{\mathbf{G}}

\def\mbI{\mathbf{I}}

\def\mbN{\mathbf{N}}

\def\mbP{\mathbf{P}}

\def\mbR{\mathbf{R}}
\def\mbS{\mathbf{S}}

\def\mbU{\mathbf{U}}
\def\mbV{\mathbf{V}}
\def\mbW{\mathbf{W}}
\def\mbX{\mathbf{X}}
\def\mbY{\mathbf{Y}}
\def\mbZ{\mathbf{Z}}

\def\vec#1{\mathrm{vec}\left(#1\right)}

\newtheorem{theorem}{Theorem}
\newtheorem{proposition}{Proposition}
\newtheorem{lemma}{Lemma}
\newtheorem{corollary}{Corollary}

\theoremstyle{definition}
\newtheorem{definition}{Definition}

\usepackage{bbold}
\usepackage{ifpdf}
\usepackage{times}
\usepackage{layout}
\usepackage{float}
\usepackage{afterpage}
\usepackage{amsmath}
\usepackage{amstext}
\usepackage{amssymb,bm}
\usepackage{latexsym}
\usepackage{color}
\usepackage{flexisym}
\usepackage{kantlipsum}
\usepackage{graphicx}
\usepackage{amsmath}
\usepackage{amsthm}
\usepackage{graphicx}

\usepackage[caption=false,font=footnotesize]{subfig}

\usepackage{booktabs}
\usepackage{multicol}

\usepackage{enumitem}

\usepackage[switch]{lineno}
\usepackage[named]{algo}
\usepackage[noend]{algpseudocode}
\usepackage{algorithm}

\makeatletter
\newcommand*{\rom}[1]{\expandafter\@slowromancap\romannumeral #1@}
\makeatother

\begin{document}
\setlength{\abovedisplayskip}{3pt}
\setlength{\belowdisplayskip}{3pt}

\title{Matrix Completion from One-Bit Dither Samples
}

\author{Arian Eamaz, \IEEEmembership{Graduate Student Member, IEEE}, Farhang Yeganegi, and Mojtaba Soltanalian, \IEEEmembership{Senior Member, IEEE} 

\vspace{-20pt}
\thanks{This work was supported in part by the National Science Foundation Grant CCF-1704401. The conference precursor to this work was presented at the 2023 IEEE International Conference on Sampling Theory and Applications (SampTA). (\emph{Corresponding author: Arian Eamaz})}
\thanks{The authors are with the Department of Electrical and Computer Engineering, University of Illinois Chicago, Chicago, IL 60607, USA (e-mail: \emph{ \{aeamaz2, fyegan2, msol\}@uic.edu}).}
}
\markboth{Submitted to the IEEE TRANSACTIONS ON SIGNAL PROCESSING, 2023}
{Shell \MakeLowercase{\textit{et al.}}: Bare Demo of IEEEtran.cls for IEEE Journals}
\maketitle

\begin{abstract}
We explore the impact of coarse quantization on matrix completion in the extreme scenario of \emph{dithered one-bit sensing}, where the matrix entries are compared with random dither levels. In particular, instead of observing a subset of high-resolution entries of a low-rank matrix, we have access to a small number of \emph{one-bit samples}, generated as a result of these comparisons. In order to recover the low-rank matrix using its coarsely quantized known entries, we begin by transforming the problem of one-bit matrix completion (one-bit MC) with random dithering into a nuclear norm minimization problem. The one-bit sampled information is represented as linear inequality feasibility constraints. We then develop the popular \emph{singular value thresholding} (SVT) algorithm to accommodate these inequality constraints, resulting in the creation of the \emph{O}ne-\emph{B}it \emph{SVT} (OB-SVT). Our findings demonstrate that incorporating multiple random dither sequences in one-bit MC can significantly improve the performance of the matrix completion algorithm. In pursuit of achieving this objective, we utilize diverse dithering schemes, namely uniform, Gaussian, and discrete dithers. To accelerate the convergence of our proposed algorithm, we introduce three variants of the OB-SVT algorithm. Among these variants is the randomized sketched OB-SVT, which departs from using the entire information at each iteration, opting instead to utilize sketched data. This approach effectively reduces the dimension of the operational space and accelerates the convergence. We perform numerical evaluations comparing our proposed algorithm with the maximum likelihood estimation method previously employed for one-bit MC, and demonstrate that our approach can achieve a better recovery performance.
\end{abstract}


\setlength{\abovedisplayskip}{3pt}
\setlength{\belowdisplayskip}{3pt}

\section{Introduction}
\label{intro}
Matrix completion, which involves recovering an unknown low-rank matrix based on partial information, is a ubiquitous challenge in numerous practical fields. As a distinctive form of low-rank matrix sensing, matrix completion presents unique difficulties, particularly because the sampling matrices may not necessarily adhere to the conditions of the matrix restricted isometry property (RIP) \cite{chi2019nonconvex}. A notable revelation has surfaced: when dealing with a matrix $\mbX$ of rank $r$, and barring excessive \emph{structure}, a compact, randomly chosen subset of its entries can facilitate an exact reconstruction. This pivotal finding was initially substantiated by \cite{candes2009exact} through a meticulous analysis of a convex relaxation technique introduced by \cite{fazel2002matrix}. One parameter that plays a crucial role in determining the
feasibility of matrix completion is a certain incoherence measure, proposed in \cite{candes2009exact}.

A prime example for matrix completion applications is the Netflix challenge, where the user rating matrix is presumed to be approximately low rank owing to the widely accepted notion that only a handful of factors influence a person's taste or preference\cite{candes2009exact,chi2019nonconvex}.
Another emerging application of matrix completion is in waveform design for multiple-input and multiple-output (MIMO) radars\cite{sun2015mimo,kalogerias2013matrix}. In a comprehensive manner, the authors of \cite{sun2015mimo,sun2015waveform,kalogerias2013matrix} conducted an extensive exploration of MIMO radars employing sparse sensing and matrix completion techniques. These approaches prove highly advantageous in significantly reducing the data volume required by MIMO radars for accurate target detection and estimation. In MIMO-MC radars, the receive antennas transmit subsampled target returns to a fusion center, which fills a data matrix partially. Subsequently, MC techniques are employed to complete the data matrix, enabling the estimation of target parameters using standard array processing methods. MIMO radars equipped with sparse sensing and matrix completion techniques have the potential to drastically decrease the amount of data necessary for precise target detection and estimation. This reduction in data volume can lead to a more efficient and accurate MIMO radar system. 
.

Quantization is a crucial step in digital-signal processing that converts continuous signals into discrete representations. However, to achieve high-resolution quantization, a large number of quantization levels are required, which can result in increased power consumption, manufacturing cost, and a reduced sampling rate in analog-to-digital converters (ADCs). To mitigate these issues, researchers have explored the use of fewer quantization bits, including the extreme case of one-bit quantization where the signals are compared with a fixed threshold at the ADCs, resulting in sign outputs. This approach enables high-rate sampling while reducing implementation cost and energy consumption compared to multi-bit ADCs. One-bit ADCs have proven to be immensely valuable across a range of applications, including MIMO \cite{mezghani2018blind}, and array signal processing \cite{liu2017one}. 

Considerable research has been undertaken to examine the effects of adding a dither before the quantization process to the signal, specifically on quantization error\cite{schuchman1964dither,carbone1997quantitative}. The mean squared error (MSE) is a relevant parameter used to evaluate the performance of a quantizing system, particularly in digital signal processing applications. By employing dithering, the system's performance can be flexibly controlled, as it enables a trade-off between accuracy and resolution. Dithering has the remarkable property of reducing the overall average quantization error \cite{carbone1997quantitative}.
Multiple dithering schemes have been studied, including uniform dither, discrete dither, and Gaussian dither. Research findings indicate that a carefully chosen dither signal can significantly enhance the resolution and performance of digital instrumentation. While discrete dithering has also been analyzed, it generally does not perform as well as uniform or Gaussian dithering in terms of quantization error reduction \cite{carbone1997quantitative}. 
It is important to acknowledge that while discrete dithering may not match the performance of uniform or Gaussian dithering, it offers notable advantages in terms of practical implementation. Specifically, in the context of Cram\'er-Rao bound (CRB) analysis, Gaussian dithering exhibits better performance for input signals with bell-shaped distributions \cite{xi2020gridless}.
The random dithers are realized through a randomly dithered generator within the ADC \cite{robinson2019analog}. The source of this random dither is a low-cost thermal noise diode, which may require additional circuitry and amplifiers to enhance the noise levels; see, for instance, \cite{ali2020background} for the implementation of multiple dithering in 12-bit, 18 gigasamples per second (GS/s) ADC.

In recent works, researchers have proposed using random dithering to improve the estimation of signal characteristics, addressing this challenge \cite{xi2020gridless}. For instance, \cite{eamaz2021modified,eamaz2022covariance,eamaz2023covariance} considered Gaussian dithering for recovering covariance from one-bit measurements. 
In \cite{dirksen2022covariance}, the authors provide a comprehensive investigation of estimating the covariance matrix from one-bit measurements using uniform dithering.
Moreover, dithered one-bit quantization has found applications in various problems discussed in contemporary literature, including sparse parameter estimation \cite{thrampoulidis2020generalized}, compressed sensing \cite{ xu2020quantized,dirksen2021non,eamaz2023harnessing}, phase retrieval \cite{eamaz2022phase}, and sampling theory \cite{eamaz2022uno}.

The theory behind matrix completion typically assumes that observations consist of continuous values within the matrix. However, the Netflix problem involves "quantized" observations that are restricted to integers between 1 and 5. This presents a challenge for high-resolution matrix completion techniques, as the impact of coarse quantization becomes more apparent. This issue is particularly prominent in recommender systems, where ratings are reduced to a single bit indicating a positive or negative rating (e.g., rating music on Pandora, determining the relevance of advertisements on Hulu, or evaluating posts on sites like MathOverflow). In such cases, the assumptions made in existing matrix completion theory do not hold true\cite{davenport20141}. MIMO radar systems with one-bit ADC receivers provide another example where traditional matrix completion techniques may not be used for waveform design due to coarse quantization applied to the received measurements \cite{sun2015mimo}.

Randomized sketching is a widely used technique for dimension reduction that encompasses various iterative methods for solving linear systems and their variations, as well as extensions to non-linear optimization problems \cite{martinsson2020randomized}. 
In this approach, instead of utilizing all available information at each iteration, only a sketched subset of the information is used, which is selected through a sketching process.  
This paper introduces a novel approach where, for the first time in the literature, the first algorithmic paradigm is employed to sketch the one-bit data matrix. The proposed method utilizes sketched information exclusively at each iteration, resulting in a more efficient structure.

\subsection{Prior Arts}
In \cite{davenport20141,bhaskar20151}, the initial attempt to address one-bit matrix completion (one-bit MC) involved developing theoretical guarantees under the generalized linear model. The authors derived a regularized maximum likelihood estimate (MLE) based on a probability distribution determined by the real-valued noisy entries of the low-rank matrix. To constrain the MLE problem, the authors employed the nuclear and Frobenius norms, drawing inspiration from previous work on one-bit compressed sensing\cite{davenport20141}. They utilized projected gradient descent to solve the regularized MLE obtained. In \cite{davenport20141} and \cite{bhaskar20151}, the authors present various theoretical guarantees and necessary conditions for achieving guaranteed recovery performance in the context of the MLE problem. They establish these guarantees by considering the rank and number of measurements required. A max-norm constrained MLE for the one-bit MC problem was comprehensively investigated in \cite{cai2013max}. Authors of \cite{ni2016optimal}, developed a greedy algorithm that extends the concept of conditional gradient descent to efficiently solve the regularized MLE for the one-bit MC problem. The concept of regularized MLE for the one-bit MC problem was further extended to the quantized MC problem in \cite{bhaskar2016probabilistic}. For quantized MC problem, the authors of \cite{cao2015categorical} considered
a Trace norm regularized MLE with a likelihood function for categorical distribution. In \cite{gao2018low}, a regularized MLE for MC from quantized and erroneous
measurements was proposed which considers the sparse additive
error in the model. 

Furthermore, \cite{davenport20141} includes a set of experiments that not only validate the implications of their theorems but also demonstrate practical applications of one-bit MC. Notably, the authors compare one-bit MC to standard matrix completion methods using movie rating data, where users submit ratings ranging from 1 to 5. To utilize one-bit MC, they quantize the data to a single bit, while the standard matrix completion algorithm has access to the original ratings. \emph{Surprisingly, the approach based on binary data outperforms the standard matrix completion approach significantly}\cite{davenport20141}.

\subsection{Motivation and Contributions of the Paper}
To derive the MLE and take advantage of random dithering, authors of \cite{davenport20141} considered noisy measurements, where the noise can be seen as random dithers. However, as demonstrated in \cite{baraniuk2017exponential,eamaz2022uno,eamaz2022phase,eamaz2022covariance,xu2020quantized}, the design of random dithers is a critical aspect of one-bit sensing that can significantly improve signal reconstruction performance. However, by utilizing noise as our dithering, as was demonstrated in \cite{davenport20141}, we constrain our dithers to follow the behavior of the noise, which is not under our control. In order for the recovery of the matrix to be feasible, certain assumptions must be made on the noise distribution, as discussed in \cite{davenport20141,cai2013max,bhaskar20151,gao2018low,ni2016optimal}. In practical scenarios, however, we cannot assume that the noise follows a distribution with properties that ensure good recovery performance for the regularized MLE problem.

Dithering is particularly relevant in matrix completion scenarios like recommendation systems, where users may prefer to compare products rather than provide exact ratings. This approach can improve the user interface's usability and, in some cases, enhance the accuracy of ratings \cite{bose2018low,davenport20141}. For example, a user may evaluate an anime's quality by comparing it to their favorite show, such as Jujutsu Kaisen, and rate it on IMDb accordingly, or simply give it a thumbs-up or thumbs-down as one-bit data.

The concept of dithered one-bit MC can be applied in digital signal processing applications where ADCs benefit from the advantages of efficient sampling in matrix completion, treating low-rate samples as partially observed data matrices, and the cost-effectiveness of coarse quantization. 
This approach holds great potential for applications such as MIMO-MC \cite{sun2015mimo} and WS-MIMO \cite{sun2019target}. 
Our main contributions in this paper are:\\
\textbf{1) Combined matrix completion with dithered one-bit quantization framework.} 
In the proposed dithered one-bit MC framework, we leverage upon the
benefits of both one-bit sensing with random dithering and efficient sampling in the matrix completion. We frame this as a nuclear norm minimization problem, constrained by a \emph{linear inequality} system obtained from the one-bit sensing with multiple dithering scheme. \emph{Our formulation allows for the freedom to design or select appropriate dithers to enhance the recovery performance.} 
\\
\textbf{2) SVT-based recovery.} 
We propose the \emph{O}ne-\emph{B}it \emph{SVT} (OB-SVT) algorithm, which utilizes an SVT specifically developed to address nuclear norm minimization problems with linear inequality constraints. To enhance the convergence rate and recovery performance, we propose additional variants of OB-SVT. Specifically, we incorporate the concept of sketch-and-project into SVT. This involves utilizing sketched one-bit measurements obtained by multiplying the sparse Gaussian sketch matrix with the one-bit data matrix for each inequality constraint. 
Our approach reduces computational load by using sketched measurements instead of processing full one-bit data from multiple comparisons, and we show its superior performance compared to the MLE method proposed by \cite{davenport20141} in both noiseless and noisy scenarios.
\\
\textbf{3) Performance guarantees.}  
Our theoretical analyses provide an upper bound on the error recovery of the one-bit MC problem that is applicable to any chosen algorithm. To achieve the recovery performance, we determine the probability and the required number of measurements. Additionally, we prove the convergence of the proposed algorithms and provide various theoretical guarantees to ensure their convergence. 
\\
\textbf{4) Recovery in the presence of additive noise.} 
The proposed OB-SVT algorithm has been slightly modified to handle the noisy one-bit MC problem. We conduct a comprehensive investigation into the performance of the proposed OB-SVT algorithms in the presence of noise. We experimentally evaluate the algorithm's performance under Gaussian and Poisson noise scenarios and observe its robustness against these types of noise. 
\\
\textbf{5) Comparing different dithering scenarios performance.}
This paper presents a comprehensive comparison of Gaussian, uniform, and discrete dithering techniques in the context of one-bit sensing for matrix completion. Notably, we introduce the use of discrete dithering. While discrete dithering may not offer the same level of performance as other schemes, its simplicity and ease of implementation enable us to generate a larger number of dithering sequences, resulting in improved recovery performance. Our numerical results show that using uniform dithering improves signal reconstruction performance when the input matrix comes from a uniform distribution. However, when measurements are corrupted by bell-shaped noise or the input follows a Gaussian distribution, Gaussian dithering performs better than uniform dithering.
\\
\textbf{6) Dithering design with adaptive scheme.}
We introduce an adaptive dither design inspired by the Bregman iterative method \cite{ma2011fixed,yin2008bregman}. 
This adaptive framework for one-bit MC, utilizes the concept of Bregman projection iterative method to iteratively design a dither that improves recovery performance. 

\subsection{Organization and Notation}


In the following section, we present an overview of one-bit quantization with random dithers. We introduce the one-bit MC problem, where the task is to reconstruct a signal from one-bit sampled measurements. We formulate this problem as a system of linear inequalities and describe the one-bit MC solver, which utilizes nuclear norm minimization for both noisy and noiseless scenarios.
Section~III introduces the OB-SVT algorithms, which are specifically designed for one-bit MC. In Section~IV, we provide various theoretical guarantees for matrix recovery from the one-bit MC problem and analyze the convergence rates of the proposed algorithms in both noisy and noiseless cases.
Section~V focuses on the adaptive thresholding process by combining the Bregman iterative method with the OB-SVT algorithms. We discuss how this approach improves the recovery performance.
In Section~VI, we present several numerical experiments to demonstrate the effectiveness of the proposed algorithms. We compare them with state-of-the-art algorithms and evaluate their performance in both noisy and noiseless scenarios using different dithering schemes.
Finally, in Section~VII, we conclude our findings and highlight the open discussion of this research.

\underline{\emph{Notation:}} Throughout this paper, we use bold lowercase and bold uppercase letters for vectors and matrices, respectively.  We represent a vector $\mathbf{x}$ and a matrix $\mbB$ in terms of their elements as $\mathbf{x}=[x_{i}]$ and $\mathbf{B}=[B_{i,j}]$, respectively. The generalized inverse operator is presented by $(\cdot)^{-}$. The sets of complex and real numbers are $\mathbb{C}$ and $\mathbb{R}$, respectively; $(\cdot)^{\top}$ and $(\cdot)^{\mathrm{H}}$ are the vector/matrix transpose and the Hermitian transpose, respectively. We define $\mathbf{x}\succeq \mathbf{y}$ as a component-wise inequality between vectors $\mathbf{x}$ and $\mathbf{y}$, i.e., $x_{i}\geq y_{i}$ for every index $i$. The function $\textrm{diag}(\mbb)$ denotes a diagonal matrix with $\{b_{i}\}$ as its diagonal elements.
The nuclear norm of a matrix $\mbB\in \mathbb{C}^{M\times N}$ is denoted $\left\|\mbB\right\|_{\star}=\sum^{r}_{i=1}\sigma_{i}$ where $r$ and $\left\{\sigma_{i}\right\}$ are the rank and singular values of $\mbB$, respectively. The maximum and minimum singular values of $\mbB$ are $\sigma_{\mathrm{max}}\left(\mbB\right)$ and $\sigma_{\mathrm{min}}\left(\mbB\right)$, respectively. The Frobenius norm of a matrix $\mathbf{B}\in \mathbb{C}^{M\times N}$ is defined as $\|\mathbf{B}\|_{\mathrm{F}}=\sqrt{\sum^{M}_{r=1}\sum^{N}_{s=1}\left|b_{rs}\right|^{2}}$, where $b_{rs}$ is the $(r,s)$-th entry of $\mathbf{B}$. The $\ell_{p}$-norm of a vector $\mathbf{b}$ is $\|\mathbf{b}\|_{p}=\left(\sum_{i}b^{p}_{i}\right)^{1/p}$. 
The Hadamard (element-wise) product is $\odot$. The vectorized form of a matrix $\mbB$ is written as $\vec{\mbB}$. Given a scalar $x$, we define the operator $(x)^{+}$ as $\max\left\{x,0\right\}$. For an event $\mathcal{E}$, $\mathbb{1}_{(\mathcal{E})}$ is the indicator function for that event meaning that $\mathbb{1}_{(\mathcal{E})}$ is $1$ if $\mathcal{E}$ occurs; otherwise, it is zero. The set $[n]$ is defined as $[n]=\left\{1,\cdots,n\right\}$. The function $\operatorname{sgn}(\cdot)$ yields the sign of its argument. The function $\log(\cdot)$ denotes the natural logarithm, unless its base is otherwise stated. The notation $x \sim \mathcal{U}(a,b)$ means a random variable drawn from the uniform distribution over the interval $[a,b]$ and $x \sim \mathcal{N}(\mu,\sigma^2)$ represents the normal distribution with mean $\mu$ and variance $\sigma^2$. A covering number is the number of $r$-balls of a given size needed to completely cover a given set $\mathcal{K}$, i.e., $\mathcal{N}\left(\mathcal{K},\|\cdot\|_2, r\right)$. The Kolmogorov $r$-entropy of a set $\mathcal{K}$ is denoted by $\mathcal{H}\left(\mathcal{K},r\right)$ defined as
the logarithm of the size of the smallest $r$-net of $\mathcal{K}$. The Hamming distance between $\operatorname{sgn}(\mbx),\operatorname{sgn}(\mby)\in\{-1,1\}^n$ 
is defined as 
\begin{equation}
\label{hamming}
d_{\mathrm{H}}(\operatorname{sgn}(\mbx),\operatorname{sgn}(\mby))=\frac{1}{n}\sum_{i=1}^{n}\mathbb{I}_{(\operatorname{sgn}(x_i)\neq \operatorname{sgn}(y_i))}.
\end{equation}
If there exists a $c>0$ such that $a\leq c b$ (resp. $a\geq c b$) for two quantities
$a$ and $b$, we have $a\lesssim b$ (resp. $a\gtrsim b$).

\section{One-Bit Matrix Completion}
\label{OB-MC}
In this section, we first present the one-bit quantization 
using multiple random dithers as a linear feasibility problem. We then formulate the one-bit MC problem as a nuclear norm minimization problem with linear inequality constraints for both noisy and noiseless scenarios. Finally, we develop the SVT algorithm in such a way to handle the linear inequality constraints and recover the low-rank matrix. We call our proposed algorithm, OB-SVT.

\subsection{Dithered One-Bit Sensing}
\label{OB}
In practice, the discrete-time samples occupy pre-determined quantized values. We denote the quantization operation on $x_k$ by the function $Q(\cdot)$. This yields the quantized signal as $r_{k} = Q(x_{k})$.
These dithers may be chosen from any distribution. 
In this paper we examine the performance of one-bit quantization under the discrete, uniform and Gaussian dithering scenarios. A Gaussian non-zero dither vector $\boldsymbol{\uptau}_{\mathcal{N}}=\left[\tau_{k}\right]$ follows the distribution $\boldsymbol{\uptau}_{\mathcal{N}} \sim \mathcal{N}\left(\mathbf{d}=\mathbf{1}d,\bSigma\right)$ with mean $d$ and the covariance matrix $\bSigma$. We also employ uniformly distributed dithers in the sequel as $\boldsymbol{\uptau}_{\mathcal{U}}\sim \mathcal{U}_{\left[a,b\right]}$.
We define the discrete dithering as $\boldsymbol{\uptau}_{\mathcal{D}}$, which contains values as $\tau_k\in\left[\frac{k}{M}D\right]$,
where $D$ represents
the peak-to-peak dither amplitude, and $k\in\{-M/ 2, \cdots,-1,1, \cdots, M / 2\}$ with an even number $M$. It is important to note that when the number $M$ of equiprobable impulses in the discrete dithering tends to infinity, it behaves similarly to a uniform distribution. The distribution of discrete dither is given by \cite{wagdy1989effect}
\begin{equation}
\label{Stephania_8}
f(\tau)=\frac{1}{M}\sum_{k}\delta\left(\tau-\frac{k}{M} D\right). 
\end{equation}
For one-bit quantization with such random dithers, $r_{k} = \operatorname{sgn}\left(x_{k}-\tau_{k}\right)$.
For notational simplicity, hereafter, we denote the random dithers by dropping the subscripts, i.e. $\boldsymbol{\uptau}=\left[\tau_{k}\right]$.

The information gathered through the one-bit sampling with random dithers may be formulated in terms of an overdetermined linear system of inequalities. We have $r_{k}=+1$ when $x_{k}>\tau_{k}$ and $r_{k}=-1$ when $x_{k}<\tau_{k}$. Collecting all the elements in the vectors as $\mathbf{x}=[x_{k}] \in \mathbb{R}^{n}$ and $\mathbf{r}=[r_{k}] \in \{-1,1\}^{n}$, therefore, one can formulate the geometric location of the signal as 
\begin{equation}
\label{eq:4}
r_{k}\left(x_{k}-\tau_{k}\right) \geq 0.
\end{equation}
Then, the vectorized representation of (\ref{eq:4}) is $\mathbf{r} \odot \left(\mathbf{x}-\boldsymbol{\uptau}\right) \succeq \mathbf{0}$
or equivalently
\begin{equation}
\label{eq:6}
\begin{aligned}
\bOmega \mathbf{x} &\succeq \mathbf{r} \odot \boldsymbol{\uptau},
\end{aligned}
\end{equation}
where $\bOmega \triangleq \operatorname{diag}\left(\mathbf{r}\right)$. Suppose $\mathbf{x},\boldsymbol{\uptau} \in \mathbb{R}^{n}$, and that $\boldsymbol{\uptau}^{(\ell)}$ denotes the random dither in $\ell$-th signal sequence, where  $\ell\in [m]$. 

For the $\ell$-th signal sequence, (\ref{eq:6}) becomes
\begin{equation}
\label{eq:7}
\begin{aligned}
\bOmega^{(\ell)} \mathbf{x} &\succeq \mathbf{r}^{(\ell)} \odot \boldsymbol{\uptau}^{(\ell)}, \quad \ell \in [m],
\end{aligned}
\end{equation}
where $\bOmega^{(\ell)}=\operatorname{diag}\left(\mathbf{r}^{(\ell)}\right)$. Denote the concatenation of all $m$ sign matrices as 
\begin{equation}
\label{eq:9}
\Tilde{\bOmega}=\left[\begin{array}{c|c|c}
\bOmega^{(1)} &\cdots &\bOmega^{(m)}
\end{array}\right]^{\top}, \quad 
\Tilde{\bOmega}\in \{-1,0,1\}^{m n\times n}.
\end{equation}
Rewrite the $m$ linear system of inequalities in  (\ref{eq:7}) as
\begin{equation}
\label{eq:8}
\Tilde{\bOmega} \mathbf{x} \succeq \operatorname{vec}\left(\mathbf{R}\right)\odot \operatorname{vec}\left(\bGamma\right),
\end{equation}
where $\mathbf{R}$ and $\bGamma$ are matrices, whose columns are the sequences $\left\{\mathbf{r}^{(\ell)}\right\}_{\ell=1}^{m}$ and $\left\{\boldsymbol{\uptau}^{(\ell)}\right\}_{\ell=1}^{m}$, respectively. 

The linear system of inequalities in (\ref{eq:8}) associated with the one-bit sampling scheme is overdetermined. We recast (\ref{eq:8}) into a \textit{one-bit polyhedron} as
\begin{equation}
\label{eq:8n}
\begin{aligned}
\mathcal{P} = \left\{\mathbf{x} \mid \Tilde{\bOmega} \mathbf{x} \succeq \operatorname{vec}\left(\mathbf{R}\right)\odot \operatorname{vec}\left(\bGamma\right)\right\}.
\end{aligned}
\end{equation}
\subsection{One-Bit MC as Nuclear Norm Minimization Problem}
Assume we apply the coarse quantization to the observed partial entries of a low-rank matrix 
$\mbX\in\mathbb{R}^{n_1\times n_2}$ of rank $r$. Define $\mathcal{P}_{\Omega}\left(\mbX\right)=\left[\widehat{X}_{i,j}\right]$ be the orthogonal projector onto the span of matrices vanishing outside of $\Omega$. These partial entries of $\mbX$ are obtained in subset $\Omega$ as below:
\begin{equation}
\label{St_1}
\begin{aligned}
\widehat{X}_{i,j} &= \begin{cases} X_{i,j} & (i,j)\in \Omega,\\ 0 & \text{otherwise}.
\end{cases}
\end{aligned}
\end{equation}
In one-bit MC, we solely observe the 
$\ell$-th one-bit data matrix $\boldsymbol{\mathcal{R}}^{(\ell)}=\left[r_{i,j}^{(\ell)}\right]\in\{-1,0,1\}^{n_1\times n_2}$, where 
the entries in $\boldsymbol{\mathcal{R}}^{(\ell)}$ are dependent on the comparison between corresponding entries in $\mathcal{P}_{\Omega}\left(\mbX\right)$ and $\ell$-th dithering matrix $\boldsymbol{\mathcal{T}}^{(\ell)}=\left[\tau_{i,j}^{(\ell)}\right]\in \mathbb{R}^{n_1\times n_2}$ according to the following relationship:
\begin{equation}
\label{St_2}
\begin{aligned}
r_{i,j}^{(\ell)} &= \begin{cases} +1 & X_{i,j}>\tau_{i,j}^{(\ell)},\\ -1 & X_{i,j}<\tau_{i,j}^{(\ell)},
\end{cases} \quad (i,j)\in\Omega, \quad \ell\in [m],
\end{aligned}
\end{equation}
and $r_{i,j}^{(\ell)}=0$, for all $(i,j)\notin\Omega,\ell\in[m]$.
Define $\mbP\in\{0,1\}^{m^{\prime}\times n_1 n_2}$ as a permutation matrix that only selects the subset $\Omega$, where $\left|\Omega\right|=m^{\prime}$. For $\ell$-th dither matrix, we formulate the obtained one-bit scheme as the following linear feasibility problem:
\begin{equation}
\label{St_3}
\mbP\bOmega^{(\ell)}\operatorname{vec}\left(\mbX\right) \succeq\mbP\left(\operatorname{vec}\left(\boldsymbol{\mathcal{R}}^{(\ell)}\right)\odot\operatorname{vec}\left(\boldsymbol{\mathcal{T}}^{(\ell)}\right)\right), \quad \ell \in [m],
\end{equation}
where $\bOmega^{(\ell)}=\operatorname{diag}\left(\operatorname{vec}\left(\boldsymbol{\mathcal{R}}^{(\ell)}\right)\right)$.
As demonstrated earlier, we rewrite \eqref{St_3} as
\begin{equation}
\label{St_4}
\begin{aligned}
\boldsymbol{\mathcal{B}} \operatorname{vec}\left(\mbX\right) &\succeq \operatorname{vec}\left(\mathbf{R}\right)\odot \operatorname{vec}\left(\bGamma\right),
\end{aligned}
\end{equation}
where $\mathbf{R}$ and $\bGamma$ are matrices, whose columns are the sequences $\left\{\mbP\operatorname{vec}\left(\boldsymbol{\mathcal{R}}^{(\ell)}\right)\right\}_{\ell=1}^{m}$ and $\left\{\mbP\operatorname{vec}\left(\boldsymbol{\mathcal{T}}^{(\ell)}\right)\right\}_{\ell=1}^{m}$, respectively, and
\begin{equation}
\label{St_5}
\boldsymbol{\mathcal{B}}=\left[\begin{array}{c|c|c}
\bOmega^{(1)}\mbP^{\top} &\cdots &\bOmega^{(m)}\mbP^{\top}
\end{array}\right]^{\top}.
\end{equation}
Therefore, to recover the low-rank matrix $\mbX$ from highly-quantized observed measurements, we must find the optimal solution from the following feasible set:
\begin{equation}
\label{St_6}
\begin{aligned}
\mathcal{F}= \left\{\mbX \mid \boldsymbol{\mathcal{B}} \operatorname{vec}\left(\mbX\right) \succeq \operatorname{vec}\left(\mathbf{R}\right)\odot \operatorname{vec}\left(\bGamma\right),~\left\|\mbX\right\|_{\star}\leq\epsilon\right\},
\end{aligned}
\end{equation}
where $\epsilon$ is the predefined threshold. The feasible set of one-bit MC is written as a \emph{nuclear norm minimization} problem as below
\begin{equation}
\label{St_7}
\begin{aligned}
\mathcal{F}^{(1)}:\quad\underset{\mbX}{\textrm{minimize}} \quad &\tau\left\|\mbX\right\|_{\star}+\frac{1}{2}\left\|\mbX\right\|^{2}_{\mathrm{F}}\\ \text{subject to} \quad &\boldsymbol{\mathcal{B}} \operatorname{vec}\left(\mbX\right) \succeq \operatorname{vec}\left(\mathbf{R}\right)\odot \operatorname{vec}\left(\bGamma\right),
\end{aligned}
\end{equation}
for some fixed $\tau\geq 0$.
The Frobenius norm is considered to control the amplitudes of the unknown data \cite{cai2010singular}.
In Section~\ref{pro}, we will propose various algorithms to tackle this problem.

\subsection{One-Bit MC With Noisy Entries}
\label{Noise}
Herein, we formulate the noisy version of one-bit MC with random dithers. Denote $\mbZ\in\mathbb{R}^{n_1\times n_2}$ as the noise matrix. The noisy one-bit samples are generated as
\begin{equation}
\label{St_8}
\begin{aligned}
r_{i,j}^{(\ell)} &= \begin{cases} +1 & X_{i,j}+Z_{i,j}>\tau_{i,j}^{(\ell)},\\ -1 & X_{i,j}+Z_{i,j}<\tau_{i,j}^{(\ell)},
\end{cases}\quad (i,j)\in\Omega,\quad\ell\in[m].
\end{aligned}
\end{equation}
Consequently, the noisy linear system of inequalities associated with \eqref{St_8} is written as
\begin{equation}
\label{St_9}   
\boldsymbol{\mathcal{B}} \left(\operatorname{vec}\left(\mbX\right)+\operatorname{vec}\left(\mbZ\right)\right) \succeq \operatorname{vec}\left(\mathbf{R}\right)\odot \operatorname{vec}\left(\bGamma\right),
\end{equation}
or equivalently,
\begin{equation}
\label{St_10}   
\boldsymbol{\mathcal{B}} \operatorname{vec}\left(\mbX\right)+\boldsymbol{\upnu}\succeq \operatorname{vec}\left(\mathbf{R}\right)\odot \operatorname{vec}\left(\bGamma\right),
\end{equation}
where $\boldsymbol{\upnu}=\boldsymbol{\mathcal{B}}\operatorname{vec}\left(\mbZ\right)$ is the noise of our system. For instance, if we consider $\operatorname{vec}\left(\mbZ\right)\sim\mathcal{N}\left(\boldsymbol{\upmu},\bSigma_z\right)$ with mean vector $\boldsymbol{\upmu}$ and covariance matrix $\bSigma_z$, the distribution of $\boldsymbol{\upnu}$ will be $\mathcal{N}\left(\boldsymbol{\mathcal{B}}\boldsymbol{\upmu}, \boldsymbol{\mathcal{B}}\bSigma_z \boldsymbol{\mathcal{B}}^{\mathrm{H}}\right)$.

To find the feasible solution from a linear feasibility problem such as $\mathbf{C}\mathbf{x}  \succeq\mathbf{b}$, we can rewrite the problem as\cite{leventhal2010randomized}
\begin{equation}
\label{St_11}
\left(\mathbf{b}-\mathbf{C}\mathbf{x}\right)^{+}=0.
\end{equation}
If our system encounters a noise vector $\mbn$, we can slightly adjust \eqref{St_11} to minimize the impact of noise using the following formula:
\begin{equation}
\label{St_12}
\left\|\left(\mathbf{b}-\mathbf{C}\mathbf{x}\right)^{+}\right\|_{\infty}\preceq \boldsymbol{\upsigma}_n,
\end{equation}
where $\boldsymbol{\upsigma}_n$ is the effect of noise. Thus, if $\mbx$ is contained within $\mathbf{C}\mathbf{x} \succeq\mathbf{b}$, there is no need to take \eqref{St_12} into account. However, if it is not, i.e. $\mbC\mbx\preceq\mbb$, we must consider the following:
\begin{equation}
\label{St_13}
\left\|\mathbf{b}-\mathbf{C}\mathbf{x}\right\|_{\infty}\preceq\boldsymbol{\upsigma}_n.
\end{equation}
Since $\mathbf{C}\mathbf{x} \preceq\mathbf{b}$, \eqref{St_13} is equivalent to
\begin{equation}
\label{St_14}
\mathbf{b}-\boldsymbol{\upsigma}_n\preceq\mathbf{C}\mathbf{x}\preceq \mbb.
\end{equation}
By applying the same process to \eqref{St_9}, the modified linear feasibility constraint is given by
\begin{equation}
\label{St_140}
\begin{aligned}
\mathcal{Y}= \left\{ \left\{\boldsymbol{\mathcal{B}}\operatorname{vec}\left(\mbX\right) \succeq \mbt\right\}\cup \left\{ \mbt-\boldsymbol{\upsigma}_z\preceq\boldsymbol{\mathcal{B}}\operatorname{vec}\left(\mbX\right) \preceq \mbt\right\}\right\},
\end{aligned}
\end{equation}
where $\mbt=\operatorname{vec}\left(\mathbf{R}\right)\odot \operatorname{vec}\left(\bGamma\right)$, and $\boldsymbol{\upsigma}_z$ is the effect of $\operatorname{vec}\left(\mbZ\right)$. Note that $\boldsymbol{\upsigma}_z$ is chosen based on the upper bound of the second norm of noise. Therefore, the nuclear norm minimization problem $\mathcal{F}^{(1)}$ is reformulated as
\begin{equation}
\label{St_15}
\begin{aligned}
\mathcal{F}^{(2)}:\quad\underset{\mbX}{\textrm{minimize}} \quad &\tau\left\|\mbX\right\|_{\star}+\frac{1}{2}\left\|\mbX\right\|^{2}_{\mathrm{F}}\\\text{subject to} \quad &\left(\mbt-\boldsymbol{\mathcal{B}}\operatorname{vec}\left(\mbX\right)\right)^{+}\preceq\boldsymbol{\upsigma}_{z}.
\end{aligned}
\end{equation}

\section{Proposed Algorithms}
\label{pro}
The main contribution in the literature on one-bit MC problem is the introduction of the regularized MLE approach \cite{davenport20141,cai2010singular,bhaskar20151,bhaskar2016probabilistic,ni2016optimal}. In this section, we present a novel perspective by considering one-bit MC as a nuclear norm minimization problem. To solve this problem, we adapt the popular algorithm SVT to its one-bit versions which we will call them \emph{O}ne-\emph{B}it \emph{SVT} (OB-SVT) algorithms. 
\subsection{OB-SVT-I}
\label{SVT}
To tackle $\mathcal{F}^{(1)}$, we employ the SVT algorithm generalized for linear inequality constraints inspired by \cite{cai2010singular} as follows: Denote a linear transformation $\mathcal{A}: \mathbb{R}^{n_1\times n_2}\rightarrow \mathbb{R}^{m m^{\prime}}$ and $\mathcal{A}^{\star}:\mathbb{R}^{m m^{\prime}}\rightarrow \mathbb{R}^{n_1\times n_2}$ as its adjoint operator. In $\mathcal{F}^{(1)}$, we have $\mathcal{A}\left(\mbX\right)=\boldsymbol{\mathcal{B}}\operatorname{vec}\left(\mbX\right)$. Then the Lagrangian for this problem is of the form
\begin{equation}
\label{St_19}
\mathcal{L}(\mbX,\boldsymbol{y})=\tau\left\|\mbX\right\|_{\star}+\frac{1}{2}\left\|\mbX\right\|^{2}_{\mathrm{F}}+\boldsymbol{y}^{\top} \left(\mbt-\mathcal{A}(\mbX)\right),
\end{equation}
where $\boldsymbol{y}$ is the Lagrangian multiplier and $\mbt=\operatorname{vec}\left(\mathbf{R}\right)\odot \operatorname{vec}\left(\bGamma\right)$.
The Karush-Kuhn-Tucker (KKT) conditions dictate that when dealing with inequality constraints, the Lagrange multiplier must be nonnegative, i.e., $\boldsymbol{y}\succeq 0$. Drawing inspiration from Uzawa's method \cite[Section~3.2]{cai2010singular}, we introduce slight modifications to the SVT algorithm to accommodate the inequality constraints and ensure satisfaction of $\boldsymbol{y}\succeq 0$ at every iteration $k$. The resulting expression is 
\begin{equation}
\label{St_20}
\left\{\begin{array}{l}
\mbX^{(k)}=\arg \min \mathcal{L}\left(\mbX, \boldsymbol{y}^{(k-1)}\right),\\
\boldsymbol{y}^{(k)}=\left(\boldsymbol{y}^{(k-1)}+\delta_k \left(\mbt-\mathcal{A}\left(\mbX^{(k)}\right)\right)\right)^{+},
\end{array}\right.
\end{equation}
where $\{\delta_k\}$ are positive step sizes.
If we consider the singular value decomposition (SVD) of $\mbX$ as $\mbX=\mbU\bSigma\mbV^{\top}$ and $\{\sigma_i\}$ as its singular values, the first step in \eqref{St_20} can be easily solved by using the singular value shrinkage operator comprehensively investigated in \cite{cai2010singular,ma2011fixed} which applies the partial SVD
to achieve the low-rank matrix structure as 
$\mathcal{D}_\tau(\mbX)=\mbU \mathcal{D}_\tau(\bSigma) \mbV^{\top},~ \mathcal{D}_\tau(\boldsymbol{\Sigma})=\operatorname{diag}\left(\left(\sigma_i-\tau\right)^{+}\right)$,
where $\tau\geq 0$ is the predefined threshold.
The SVT algorithm is a popular approach for matrix sensing (here matrix completion), where the partial SVD can be numerically calculated using Krylov subspace methods (e.g. Lanczos algorithm)
\cite{van1996matrix}. Consequently, 
the update process \eqref{St_20} is rewritten as 
\begin{equation}
\label{St_22}
\left\{\begin{array}{l}
\mbX^{(k)}=\mathcal{D}_\tau\left(\mathcal{A}^{\star}\left(\boldsymbol{y}^{(k-1)}\right)\right),\\
\boldsymbol{y}^{(k)}=\left(\boldsymbol{y}^{(k-1)}+\delta_k \left(\mbt-\mathcal{A}\left(\mbX^{(k)}\right)\right)\right)^{+}.
\end{array}\right.
\end{equation} 
For the noisy scenario \eqref{St_15},
the Lagrangian formulation for
$\mathcal{F}^{(2)}$ is given by
\begin{equation}
\label{St_190}
\begin{aligned}
\mathcal{L}^{n}(\mbX, \boldsymbol{y}_1)=\tau\left\|\mbX\right\|_{\star}+\frac{1}{2}\left\|\mbX\right\|^{2}_{\mathrm{F}}+\boldsymbol{y}^{\top}_1 \left(\left(\mbt-\mathcal{A}(\mbX)\right)^{+}-\boldsymbol{\upsigma}_{z}\right),
\end{aligned}
\end{equation}
where $\boldsymbol{y}_1$ is the Lagrangian multiplier parameter. We call this algorithm \emph{OB-SVT-I}. 

As outlined in Subsection~\ref{Noise}, incorporating the noisy inequality constraint \eqref{St_12} can be achieved by only considering $\boldsymbol{\mathcal{B}}\operatorname{vec}\left(\mbX^{(k)}\right)\preceq \mbt$ at each iteration $k$. Otherwise, i.e., $\boldsymbol{\mathcal{B}}\operatorname{vec}\left(\mbX^{(k)}\right)\succeq \mbt$, there is no need to update the Lagrangian multiplier $\boldsymbol{y}_1$. As a result, the operator $(\cdot)^{+}$ can be easily removed from \eqref{St_190}:
\begin{equation}
\label{St_1900}
\begin{aligned}
\mathcal{L}^{n}(\mbX, \boldsymbol{y}_1)=\tau\left\|\mbX\right\|_{\star}+\frac{1}{2}\left\|\mbX\right\|^{2}_{\mathrm{F}}+\boldsymbol{y}^{\top}_1 \left(\mbt-\boldsymbol{\upsigma}_{z}-\mathcal{A}(\mbX)\right).
\end{aligned}
\end{equation}
Denote $\mbt_1=\mbt-\boldsymbol{\upsigma}_{z}$. According to Uzawa's method, the following update process is proposed to address $\mathcal{F}^{(2)}$:
\begin{equation}
\label{Steph_10}
\left\{\begin{array}{l}
\mbX^{(k)}=\mathcal{D}_\tau\left(\mathcal{A}^{\star}\left(\boldsymbol{y}_1^{(k-1)}\right)\right),\\
\boldsymbol{g}^{(k)}=\left(\mbt_1-\mathcal{A}\left(\mbX^{(k)}\right)\right)^{+}\mathbb{1}_{\left(\mathcal{A}\left(\mbX^{(k)}\right)\preceq \mbt\right)},\\
\boldsymbol{y}^{(k)}_1= \left(\boldsymbol{y}^{(k-1)}_1+\delta_k \boldsymbol{g}^{(k)}\right)^{+}.
\end{array}\right.
\end{equation} 
As evident in the OB-SVT-I algorithm, there is no guarantee in terms of consistency of the solution,
meaning that the solution may not satisfy all linear inequalities. In order to ensure that the solution adheres to all inequalities, inspired by linear inequality solvers, we will introduce the OB-SVT-II. This variant will apply the operator $(\cdot)^{+}$ to linear constraints at each iteration, updating the solution only when it fails to satisfy the linear constraints and continuing the process until it does.
\subsection{OB-SVT-II}
By drawing inspiration from the Kaczmarz algorithm \cite{leventhal2010randomized}, where 
inequality constraints are 
represented by
\begin{equation}
\left(\mbt-\boldsymbol{\mathcal{B}}\operatorname{vec}\left(\mbX\right)\right)^{+}=0,  
\end{equation}
we can reframe the Lagrangian 
formulation associated with $\mathcal{F}^{(1)}$ as follows:
\begin{equation}
\label{Steph190}
\mathcal{L}^{(1)}(\mbX, \boldsymbol{y})=\tau\left\|\mbX\right\|_{\star}+\frac{1}{2}\left\|\mbX\right\|^{2}_{\mathrm{F}}+\boldsymbol{y}^{\top} \left(\mbt-\mathcal{A}(\mbX)\right)^{+}.
\end{equation}
The intriguing aspect of this new formulation is that the inequality constraint behaves akin to an equality constraint, removing the necessity for a 
nonnegative constraint on $\mby$. As a result, the algorithm designed to address this Lagrangian problem may be expressed as follows: 
\begin{equation}
\label{St_220}
\left\{\begin{array}{l}
\mbX^{(k)}=\mathcal{D}_\tau\left(\mathcal{A}^{\star}\left(\boldsymbol{y}^{(k-1)}\right)\right),\\
\boldsymbol{y}^{(k)}=\boldsymbol{y}^{(k-1)}+\delta_k \left(\mbt-\mathcal{A}\left(\mbX^{(k)}\right)\right)^{+}.
\end{array}\right.
\end{equation} 
This algorithm is named \emph{OB-SVT-II}, bears resemblance to both the least squares and Kaczmarz algorithms (with an identity sketch matrix) \cite{polyak1964gradient,leventhal2010randomized}, with the additional inclusion of a nuclear norm constraint in each iteration.
\subsection{Randomized OB-SVT}
The key advantage of randomized algorithms over regular algorithms lies in their ability to operate on batches or even single measurement per iteration, rather than requiring the entire set of measurements. This characteristic makes them highly efficient for practical applications. This reduction in complexity greatly enhances the efficiency of this algorithm in large-scale problems such as MC.

The new formulation of one-bit MC in \eqref{Steph190} opens up avenues for exploring the feasibility of solving the following Lagrangian problem when only partial
measurements per iteration 
are utilized:
\begin{equation}
\label{MySt❤️❤️}
\mathcal{L}^{(2)}(\mbX, \boldsymbol{y}_s)=\tau\left\|\mbX\right\|_{\star}+\frac{1}{2}\left\|\mbX\right\|^{2}_{\mathrm{F}}+\boldsymbol{y}^{\top}_s \left(\mbS\mbt-\mbS\boldsymbol{\mathcal{B}}\operatorname{vec}\left(\mbX\right)\right)^{+},
\end{equation}
where $\mbS\in\mathbb{R}^{s\times m m^{\prime}}$ is a sketch matrix with $s<m^{\prime}$ as the sketch length, which randomly selects a row or a block of matrix $\boldsymbol{\mathcal{B}}\in\{-1,0,1\}^{m m^{\prime}\times n_1 n_2}$. One well-known example of a sketch matrix is the sparse Gaussian matrix, which selects a block of the matrix using the following definition \cite{derezinski2022sharp}:
\begin{equation}
\label{MMYST}
\mbS^{\top}=\left[\begin{array}{c|c|c}
\mathbf{0}_{s\times p} &\mbG &\mathbf{0}_{s\times ((m-1) m^{\prime}-p)}
\end{array}\right]^{\top},
\end{equation}
where $\mbG$ is a $s\times m^{\prime}$ Gaussian matrix, whose entries are i.i.d. following the distribution $\mathcal{N}\left(0,1\right)$, and $p=m^{\prime}\alpha,~\alpha\in[m-1]$. 
Assume $\boldsymbol{y}_s\in \mathbb{R}^{s}$ and
$\mathcal{M}\left(\mbX\right)=\mbS\boldsymbol{\mathcal{B}}\operatorname{vec}\left(\mbX\right)$. Then the \emph{randomized OB-SVT} algorithm is written as
\begin{equation}
\label{ST❤️❤️}
\left\{\begin{array}{l}
\mbX^{(k)}=\mathcal{D}_\tau\left(\mathcal{M}^{\star}\left(\boldsymbol{y}^{(k-1)}_s\right)\right),\\
\boldsymbol{y}^{(k)}_s=\boldsymbol{y}^{(k-1)}_s+\delta_k \left(\mbS\mbt-\mathcal{M}\left(\mbX^{(k)}\right)\right)^{+}.
\end{array}\right.
\end{equation} 

\subsection{Adaptive Thresholding With Bregman Iterative Method}
\label{Adaptive_T}
We will employ the concept of 
\emph{Bregman iterative method} to devise a random dither that enhances the performance of one-bit low-rank matrix sensing. We will begin by presenting an overview of the Bregman iterative method. 

Define a convex function $g(\cdot)$, the Bregman distance of the vector $\mbu$ from the vector $\mbv$ is defined as
\begin{equation}
\label{St_23}
D_g^{\mbp}(\mbu,\mbv)\triangleq g(\mbu)-g(\mbv)- \mbp^{\top} (\mbu-\mbv),
\end{equation}
where $\mbp \in \partial g(\mbv)$ is some subgradient in the subdifferential of $g$ at the point $\mbv$. Authors of \cite{yin2008bregman} proposed solving the $\ell_1$ minimization problem for $\mbA\mbx=\mbb$, 
by replacing 
$g(\mbx)=\|\mbx\|_1$ with its Bregman distance at each iteration as follows:
\begin{equation}
\label{St_24}
\left\{\begin{array}{l}
\mbx^{(k+1)}= 
\arg \min D_g^{\mbp^{(k)}}\left(\mbx, \mbx^{(k)}\right)+\frac{1}{2}\left\|\mbA\mbx-\mbb\right\|_2^2,\\
\mbp^{(k+1)}=\mbp^{(k)}+\mbA^{\top}\left(\mbb-\mbA\mbx^{(k+1)}\right),
\end{array}\right.
\end{equation} 
which is interestingly equivalent to 
\begin{equation}
\label{St_25}
\left\{\begin{array}{l}
\mbb^{(k+1)}=\mbb+\left(\mbb^{(k)}-\mbA\mbx^{(k)}\right),\\
\mbx^{(k+1)}= 
\arg \min~\|\mbx\|_1+\frac{1}{2}\left\|\mbA\mbx-\mbb^{(k+1)}\right\|_2^2.
\end{array}\right.
\end{equation} 
This approach 
has been extended to the low-rank matrix sensing and matrix completion 
in \cite{ma2011fixed} as 
\begin{equation}
\label{St_26}
\left\{\begin{array}{l}
\mbb^{(k+1)}= 
\mbb+\left(\mbb^{(k)}-\mathcal{A}\left(\mbX^{(k)}\right)\right),\\
\mbX^{(k+1)}= 
\arg\min \tau\|\mbX\|_{\star}+\frac{1}{2}\left\|\mathcal{A}(\mbX)-\mbb^{(k+1)}\right\|_2^2.
\end{array}\right.
\end{equation} 
The study conducted in \cite{ma2011fixed} demonstrated that by utilizing the Bregman iterative method in conjunction with the nuclear norm minimization problem, both the convergence rate and accuracy of the algorithm can be improved.

To tailor the Bregman iterative method to the one-bit MC problem, it appears sufficient to incorporate the impact of inequality constraints, i.e., $\mathcal{A}\left(\mbX\right)\succeq\mbb$, in the following manner:
\begin{equation}
\label{St_27}
\left\{\begin{array}{l}
\mbb^{(k+1)}= 
\mbb+\left(\mbb^{(k)}-\mathcal{A}\left(\mbX^{(k)}\right)\right)^{+},\\
\mbX^{(k+1)}= 
\arg\min \tau\|\mbX\|_{\star}+\frac{1}{2}\left\|\left(\mbb^{(k+1)}-\mathcal{A}\left(\mbX\right)\right)^{+}\right\|_2^2.
\end{array}\right.
\end{equation} 
Also as discussed in \cite{ma2011fixed}, the second step of \eqref{St_27} can be tackled by the SVT algorithm. Consequently, the Bregman iterative method for the one-bit MC problem is formulated as 
\begin{equation}
\label{St_28}
\left\{\begin{array}{l}
\mbt^{(k+1)}= 
\mbt+\left(\mbt^{(k)}-\boldsymbol{\mathcal{B}}\operatorname{vec}\left(\mbX^{(k)}\right)\right)^{+},\\\mbX^{(k+1)}=
\arg\min \tau\|\mbX\|_{\star}+\frac{1}{2}\left\|\left(\mbt^{(k+1)}-\boldsymbol{\mathcal{B}}\operatorname{vec}\left(\mbX\right)\right)^{+}\right\|_2^2,
\end{array}\right.
\end{equation} 
where $\tau$ is a predefined threshold and $\mbt^{(k)}=\operatorname{vec}\left(\mathbf{R}\right)\odot \operatorname{vec}\left(\bGamma^{(k)}\right)$. After setting the current dither as fixed, we can obtain $\mbt=\operatorname{vec}\left(\mathbf{R}\right)\odot \operatorname{vec}\left(\bGamma\right)$ and subsequently update the dither at each iteration $k$ using the first step of \eqref{St_28}. This process is repeated until $\left\|\mbt^{(k+1)}-\mbt^{(k)}\right\|_2\leq \varepsilon$, where $\varepsilon$ is a positive constant value. Note that the second step of \eqref{St_28} is tackled by the SVT with inequality constraints proposed in Section~\ref{pro}.

\section{Theoretical Discussion}
Our focus is on the possibility of performing one-bit sensing for matrix completion while establishing an upper bound on the recovery error, regardless of the specific algorithm employed. Previous works, such as \cite{davenport20141,cai2013max}, have provided theoretical guarantees by considering the distribution of noise and the KL-divergence for the regularized MLE problem. However, in this paper, we aim to derive a more general approach that can be applied to one-bit MC, irrespective of the noise distribution.
Subsequently, we delve into the convergence analysis of the proposed OB-SVT algorithms in both noisy and noiseless scenarios.
\subsection{One-Bit MC Theoretical Guarantees}
\label{theory}
In the work presented by the authors in \cite{candes2010matrix}, noisy matrix completion is formulated as a nuclear norm minimization problem. This approach has led to the derivation of rigorous theoretical guarantees, further advancing the understanding and development of this intriguing field. The quantized measurements can be expressed as follows:
\begin{equation}
\label{Steph2}
\mathcal{P}_{\Omega}\left(\mbY\right) = \mathcal{P}_{\Omega}\left(\mbX\right)+\mbN,
\end{equation}
where the matrix $\mbN\in\mathbb{R}^{n_1\times n_2}$ presents the effect of quantization as the additive noise matrix. Therefore, the nuclear norm minimization problem 
associated with the quantized MC is given by
\begin{equation}
\label{Steph3}
\begin{aligned}
&\underset{\mbX}{\textrm{minimize}}\quad \left\|\mbX\right\|_{\star}\\
&\text{subject to} \quad \left\|\mathcal{P}_{\Omega}(\mbX-\mbY)\right\|_{\mathrm{F}} \leq q,
\end{aligned}
\end{equation}
where the parameter $q$ denotes the impact of the quantization process.
The set of
linear constraints in
\eqref{St_7} can be expressed as the following system of equations, taking into account the slack vector $\boldsymbol{\mathcal{E}}$:
\begin{equation}
\label{kvm_2}
\mathcal{A}\left(\mbX\right)=\operatorname{vec}\left(\mathbf{R}\right)\odot \operatorname{vec}\left(\bGamma\right)+\boldsymbol{\mathcal{E}}.
\end{equation}
Therefore, the constraint of nuclear norm minimization problem for the one-bit MC can be rewritten as $\left\|\mathcal{A}\left(\mbX\right)-\operatorname{vec}\left(\mathbf{R}\right)\odot \operatorname{vec}\left(\bGamma\right)\right\|_{2}\leq q$. 
As the norm-$2$ can be constrained by the norm-$1$, the derived constraint can be alternatively expressed as $\left\|\mathcal{A}\left(\mbX\right)-\operatorname{vec}\left(\mathbf{R}\right)\odot \operatorname{vec}\left(\bGamma\right)\right\|_{1}\leq q^{\prime}$. In the rest of this section, our aim is to theoretically investigate the parameter $q^{\prime}$ and establish its upper bound. Define the following operator:
\begin{definition}
\label{def_1}
For a matrix $\mbX=[X_{i,j}]\in\mathbb{R}^{n_1\times n_2}$ and the $\ell$-th dither matrix $\boldsymbol{\mathcal{T}}^{(\ell)}=\left[\tau_{i,j}^{(\ell)}\right]\in \mathbb{R}^{n_1\times n_2}$, define $d_{i,j}^{(\ell)}=\left|X_{i,j}-\tau_{i,j}^{(\ell)}\right|$ as a distance between the $(i,j)$-th entries of $\mbX$ and $\boldsymbol{\mathcal{T}}^{(\ell)}$ for $(i,j)\in\Omega$ and $\ell=[m]$. Then, we denote the average of such distances by
\begin{equation}
\label{a_3}
\begin{aligned}
T_{\mathrm{ave}}(\mbX) &= \frac{1}{mm^{\prime}}\left\|\mathcal{A}\left(\mbX\right)-\operatorname{vec}\left(\mathbf{R}\right)\odot \operatorname{vec}\left(\bGamma\right)\right\|_{1}, \\&=\frac{1}{mm^{\prime}}\sum_{(i,j)\in\Omega}\sum_{\ell=1}^{m}\left|X_{i,j}-\tau_{i,j}^{(\ell)}\right|,
\end{aligned}
\end{equation}
where $m^{\prime}=|\Omega|$. 
\end{definition}
It is important to note that the guarantee is obtained under the uniform dithering scheme. In the following definition, we state the consistent reconstruction property which will be our assumption in the provided theorem:
\begin{definition}
\label{def_2}
Define a low-rank matrix as $\mbX=[X_{i,j}]\in\mathbb{R}^{n_1\times n_2}$ and the $\ell$-th dither matrix by $\boldsymbol{\mathcal{T}}^{(\ell)}=\left[\tau_{i,j}^{(\ell)}\right]\in \mathbb{R}^{n_1\times n_2}$ for $\ell\in[m]$. Denote $\bar{\mbX}=[\bar{X}_{i,j}]\in\mathbb{R}^{n_1\times n_2}$ as a solution obtained by an arbitrary reconstruction algorithm addressing the problem \eqref{St_7}. Then, such a reconstruction algorithm is said to be consistent when
\begin{equation}
\label{a_1}
\operatorname{sgn}\left(X_{i,j}-\tau_{i,j}^{(\ell)}\right)=\operatorname{sgn}\left(\bar{X}_{i,j}-\tau_{i,j}^{(\ell)}\right),\quad (i,j)\in\Omega,~\ell\in[m],
\end{equation}
or in the matrix form
\begin{equation}
\label{a_2}
\operatorname{sgn}\left(\mathcal{P}_{\Omega}\left(\mbX-\boldsymbol{\mathcal{T}}^{(\ell)}\right)\right)=\operatorname{sgn}\left(\mathcal{P}_{\Omega}\left(\bar{\mbX}-\boldsymbol{\mathcal{T}}^{(\ell)}\right)\right),\quad\ell\in[m].
\end{equation}
\end{definition}
In the following theorem, we present the uniform recovery performance of the one-bit MC:
\begin{theorem}
\label{thr_1}
Consider the set 
\begin{equation}
\label{n_2}
\mathcal{K}_r=\left\{\mbX\in\mathbb{R}^{n_1\times n_2}\mid\operatorname{rank}(\mbX)\leq r,  \|\mbX\|_{\mathrm{max}}\leq\alpha\right\}.
\end{equation}
Suppose $m^{\prime}$ entries of $\mbX$ with locations sampled uniformly at random are compared with $m$ sequences of uniform thresholds generated as $\tau^{(\ell)}_{i,j}\sim \mathcal{U}_{\left[-\alpha,\alpha\right]}$ for all $(i,j)\in\Omega,\ell\in[m]$, resulting in the observed one-bit data. With a positive constant $c$,
and a probability of at least $1-2e^{-\frac{\epsilon^2 mm^{\prime}}{4\alpha^2}}$, the following upper recovery bound holds for all matrices $\mbX\in\mathcal{K}_r$ and all $\bar{\mbX}$ satisfying the consistent reconstruction property in Definition~\ref{def_2}:
\begin{equation}
\label{a_5}
\|\mbX-\bar{\mbX}\|_{\mathrm{F}}\leq 4\sqrt{\epsilon\alpha n_1n_2},
\end{equation}
where the required number of samples must satisfy
\begin{equation}
\label{samplereq}
mm^{\prime}\gtrsim\sqrt{r}\max\left(n_1,n_2\right).
\end{equation}
\end{theorem}
The proof of Theorem~\ref{thr_1} is comprehensively presented in Appendix~\ref{Stephania_1}. 
\begin{corollary}
\label{col_1}
The rate of reduction for $\epsilon$ concerning $m^{\prime}$ is characterized by $\mathcal{O}\left((m m^{\prime})^{-\frac{2}{5}}\right)$.   
\end{corollary}
\begin{proof}
As shown in the proof of Theorem~\ref{thr_1}, the number of measurements must meet the following bound:
$mm^{\prime}\geq \frac{2\alpha^{\frac{5}{2}}(n_1+n_2)r}{\epsilon^{\frac{5}{2}}}$.   
The lower bound of $m m^{\prime}$ depends on $\varepsilon^{-\frac{5}{2}}$, and consequently, the reduction rate of $\varepsilon$ with respect to $mm^{\prime}$ is given by $\mathcal{O}\left((mm^{\prime})^{-\frac{2}{5}}\right)$. 
\end{proof}
Note that Theorem~\ref{thr_1} remains valid even in the case of a single comparison process or $m=1$. However, if the quantization system allows for the possibility of multiple comparisons (where $m>1$), it can result in an increased number of one-bit samples without a corresponding increase in the number of high-resolution measurements.

As discussed in Section~\ref{SVT}, the obtained solution by the OB-SVT-I algorithm may not satisfy the consistent reconstruction property in Definition~\ref{def_2}. In the following theorem, we provide the upper recovery bound for such algorithms in the context of one-bit MC problem:
\begin{theorem}
\label{thr_2}
Under assumptions of Theorem~\ref{thr_1}, if a solution $\bar{\mbX}$ does not meet the consistent reconstruction property in Definition~\ref{def_2}, the following upper recovery bound holds for all $\mbX\in\mathcal{K}_r$ and all $\bar{\mbX}$ with a probability of at least $1-2e^{-\frac{\epsilon^2 mm^{\prime}}{4\alpha^2}}$:
\begin{equation}
\label{noc}
\|\mbX-\bar{\mbX}\|_{\mathrm{F}}\leq 4\sqrt{\epsilon\alpha n_1n_2}+2\alpha\sqrt{2n_1n_2d_{\mathrm{H}}\left(\operatorname{vec}\left(\mbR\right),\operatorname{vec}\left(\bar{\mbR}\right)\right)},
\end{equation} 
where $\mbR$ and $\bar{\mbR}$ are the one-bit data obtained from the data $\mathcal{P}_{\Omega}\left(\mbX\right)$ and the solution $\mathcal{P}_{\Omega}\left(\bar{\mbX}\right)$, respectively.
\end{theorem}
The proof of Theorem~\ref{thr_2} is presented in Appendix~\ref{proof_thr_2}. As can be observed in Theorem~\ref{thr_2}, inconsistency of the solution results in an increase in the upper recovery bound compared to the result presented in Theorem~\ref{thr_1} for consistent reconstruction algorithms. Consequently, it is anticipated that the error in the solution obtained by OB-SVT-II is smaller than that of OB-SVT-I. This phenomenon will be demonstrated through numerical experiments in Section~\ref{numerical}.
\subsection{Convergence Analysis of Proposed Algorithms}
\label{conv1}
The OB-SVT-I algorithm is proven to converge according to the following theorem:
\begin{theorem}
\label{ST_theo1}
Assume $m$ to be the number of random dither sequences and $\delta_k $ be the step size. The OB-SVT-I algorithm, employed to solve the nuclear norm minimization problem $\mathcal{F}^{(1)}$, exhibits convergence as follows:
\begin{equation}
\label{MySt_forever}
\left\|\boldsymbol{y}^{(k)}-\boldsymbol{y}\right\|^2_2\leq \left\|\boldsymbol{y}^{(k-1)}-\boldsymbol{y}\right\|^2_2-\left(2\delta_k-\delta^2_k m\right)\left\|\mbX^{(k)}-\mbX\right\|^2_{\mathrm{F}},
\end{equation}
where $\left(2\delta_k-\delta^2_k m\right)$ is positive. 
\end{theorem}
We thoroughly study the proof of Theorem~\ref{ST_theo1} in Appendix~\ref{Stephania_2}.
Based on the aforementioned theorem, the step size $\delta_k$ should satisfy $0 \leq \delta_k < 2/m$. This condition implies that the error between the variable and its desired value is a non-increasing function with respect to the iteration, providing evidence of the convergence of OB-SVT-I.

The convergence of the OB-SVT-II algorithm is examined through the following theorem. By establishing a convergence rate based on this theorem, we can make a meaningful comparison of its performance relative to that of the OB-SVT-I algorithm.
\label{conv2}
\begin{theorem}
\label{ST_theo2}
Let $m$ denote the number of random dither sequences, and $\delta_k$ represent the step size. The OB-SVT-II algorithm, utilized for solving the nuclear norm minimization problem $\mathcal{F}^{(1)}$, demonstrates convergence as follows:
\begin{equation}
\label{MySt_forever1}
\left\|\boldsymbol{y}^{(k)}-\boldsymbol{y}\right\|^2_2\leq \left\|\boldsymbol{y}^{(k-1)}-\boldsymbol{y}\right\|^2_2-m\delta^2_k\left\|\mbX^{(k)}-\mbX\right\|^2_{\mathrm{F}}.
\end{equation}
\end{theorem}
The detailed proof of Theorem~\ref{ST_theo2} is available in Appendix~\ref{Stephania_3}. As mentioned in Theorem~\ref{ST_theo1}, to ensure that OB-SVT-I converges, the term $2\delta_k-m\delta_k^2$ should be positive which implies that the step size $\delta_k$ must be lower than $2/m$.
However, Theorem~\ref{ST_theo2} for OB-SVT-II demonstrates that the choice of step size has more flexibility since the term $m\delta_k^2$ is 
always 
a positive value. Moreover, it is evident from 
\eqref{MySt_forever1} that by increasing the number of random dithers $m$ with a fixed step size $\delta_k$, the convergence of OB-SVT-II will be accelerated, which aligns with our goal. We will empirically examine and analyze the performance of both algorithms, demonstrating that this simple difference in step size enhances the convergence and performance of OB-SVT-II compared to OB-SVT-I.

In the following theorem, we derive the convergence rate of the randomized OB-SVT algorithm:
\begin{theorem}
\label{ST_theo3}
Assume $m$ to be the number of random dither sequences and $\delta_k$ be the step size. Let $\mbS\in\mathbb{R}^{s\times m m^{\prime}}$ defined in \eqref{MMYST} be the sparse Gaussian sketch matrix and $\alpha=c\sqrt{\log s}$ for a positive constant $c$.
The randomized OB-SVT algorithm with 
the sketch matrix $\mbS$,
employed to solve the nuclear norm minimization problem $\mathcal{F}^{(1)}$, exhibits convergence as follows:
\begin{equation}
\label{MySt_forever2}  
\mathbb{E}_{\mbS}\left\{\left\|\boldsymbol{y}^{(k)}_s-\boldsymbol{y}_s\right\|^2_2\right\} \leq \left\|\boldsymbol{y}^{(k-1)}_s-\boldsymbol{y}_s\right\|^2_2-\alpha^2\delta^2_k\left\|\mbX^{(k)}-\mbX\right\|^2_{\mathrm{F}}.
\end{equation}
\end{theorem}
The complete proof of Theorem~\ref{ST_theo3} is presented in Appendix~\ref{Stephania_4}.
The connection between the convergence of the Lagrangian multiplier $\boldsymbol{y}^{(k)}$ and the convergence of the solution $\mbX^{(k)}$ remains consistent across all suggested algorithms, as expressed in the following lemma:
\begin{lemma}
\label{St_lemma}
Given that $m$ represents the number of random dither sequences, the solution attained by OB-SVT algorithms converges towards the optimum, as described by the following relationship: 
\begin{equation}
\label{ST-AR}
\left\|\mbX^{(k)}-\mbX\right\|^2_{\mathrm{F}} \leq \frac{1}{m}  \left\|\boldsymbol{y}^{(k-1)}-\boldsymbol{y}\right\|^2_{2},
\end{equation}
where the Lagrangian multipliers converge according to Theorems~\ref{ST_theo1}-\ref{ST_theo3}.
\end{lemma}
We extensively provide the proof of Lemma~\ref{St_lemma} in Appendix~\ref{Stephania_5}. In the context of randomized OB-SVT, in order to align with the algorithm's randomized nature and establish a diminishing upper bound for the error, we take the expectation from both sides, yielding: $\mathbb{E}\left\{\left\|\mbX^{(k)}-\mbX\right\|^2_{\mathrm{F}}\right\} \leq \frac{1}{m}  \mathbb{E}\left\{\left\|\boldsymbol{y}^{(k-1)}_s-\boldsymbol{y}_s\right\|^2_{2}\right\}$.
The mentioned lemma reaffirms that increasing the quantity of random dither sequences accelerates the convergence rate of OB-SVT algorithms.
\subsection{Robustness of Proposed Algorithms in the Presence of Noise}
In this study, we examine the impact of noise or perturbations on the performance of proposed OB-SVT algorithms. We apply these algorithms to solve the problem of noisy one-bit MC and explore the condition under which all algorithms still converge and provide satisfactory recovery performance. This investigation aims to demonstrate the effectiveness and robustness of the proposed algorithms, even in the presence of noise at specific levels. 


In the context of the noisy one-bit MC problem as defined in Section~\ref{Noise}, the proposition below articulates the sufficient condition on the \emph{pre-quantization error} $\mbZ\in\mathbb{R}^{n_1\times n_2}$ to guarantee the convergence of the proposed algorithms:

\begin{proposition}
\label{prop_1}
Let $m$ denote the number of random dither sequences, and $\delta_k$ represent the step size.
For the noisy one-bit MC problem with $\mbZ\in\mathbb{R}^{n_1\times n_2}$ representing the pre-quantization error and $I$ signifying the total count of iterations at which the algorithm achieves complete convergence, the convergence of proposed algorithms is contingent upon the noise satisfying the following condition:
\begin{equation}
\label{Stephanie_1}
\left\|\mbZ\right\|^2_{\mathrm{F}}\leq \frac{1}{2}\left\|\mbX^{(I)}-\mbX\right\|^2_{\mathrm{F}}.
\end{equation}
\end{proposition}
A comprehensive demonstration of the proof for Proposition~\ref{prop_1} can be found in Appendix~\ref{ghadboland_1}. Irrespective of whether $\mbZ$ is deterministic or random, it is imperative for its Frobenius norm to adhere to the previously stated upper bound to ensure reliable convergence. Moreover, drawing insights from the aforementioned proposition, it becomes evident that in contrast to noise-free situations, the recovery error of OB-SVT algorithms in noisy scenarios is constrained by a lower bound contingent upon the impact of noise.

\section{Numerical Investigations}
\label{numerical}
In this section, we numerically scrutinize the efficacy of our proposed OB-SVT methods by comparing their recovery results with the state-of-the-art method (MLE approach) proposed in \cite{davenport20141}. Note that in all experiments the number of random dither sequences is considered to be $m=1$ except in Fig.~\ref{figure_1}f.
In particular, we constructed a random $500\times 500$ matrix $\mbX=[X_{i,j}]$ with rank $r=10$ by forming $\mbX=\mbX_{1}\mbX_{2}^{\top}$, where $\mbX_{1}$ and $\mbX_{2}$ are $500\times 10$ matrices with entries drawn i.i.d. from the Gaussian distribution $\mathcal{N}\left(0,1\right)$. In the results reported in Fig.~\ref{figure_1}a, the high-resolution measurements $\mathcal{P}_{\Omega}\left(\mbX\right)$ are corrupted by an additive noise $\mbZ=[z_{i,j}]$ following the distribution $z_{i,j}\sim\mathcal{N}(0,1)$ resulting in $\mathcal{P}_{\Omega}\left(\mbY\right)=\mathcal{P}_{\Omega}\left(\mbX+\mbZ\right)$. We then obtained one-bit observations by comparing the corrupted high-resolution values with the generated dithers $\boldsymbol{\uptau}$ and recording the sign of the resulting value. Herein, we have utilized the Gaussian 
distributed random dithers. Accordingly, inspired by \cite{10206479}, we have generated the Gaussian dithers as $\left\{\boldsymbol{\uptau}^{(\ell)}\sim\mathcal{N}\left(\mathbf{0},\frac{\beta_{\mbY}^{2}}{9}\mbI\right)\right\}_{\ell=1}^{m}$, where $\beta_{\mbY}$ denotes the 
dynamic range (DR) of $\mathcal{P}_{\Omega}\left(\mbY\right)$.
For simplicity in OB-SVT methods, we have utilized step sizes that are independent of the iteration count; i.e. $\delta_{k}=\delta$ for all $k$. Inspired from \cite{cai2010singular}, we have set the value of $\delta$ to $\delta=1.2\frac{n_{1}n_{2}}{m^{\prime}}$. The reason behind this selection has been heuristically justified in \cite{cai2010singular}. The parameter $\tau$ is chosen empirically and set to $\tau=\alpha\sqrt{n_{1}n_{2}}$, where the value of $\alpha$ varies for different values of $m^{\prime}$. Define the relative error of the reconstruction as
$\text{relative error}\triangleq\frac{\left\|\bar{\mbX}-\mbX\right\|_{\mathrm{F}}}{\left\|\mbX\right\|_{\mathrm{F}}}$,
where $\mbX$ is the real unknown matrix and $\bar\mbX$ denotes the reconstructed version of $\mbX$ by either OB-SVT or MLE methods. In Fig.~\ref{figure_1}a, the accuracy of input matrix reconstruction using OB-SVT-I, OB-SVT-II and MLE methods is compared in terms of relative error.
It can be seen that OB-SVT-I and OB-SVT-II outperform MLE in reconstructing the input matrix. The reason behind the better performance of OB-SVT-II compared to that of OB-SVT-I has been discussed in Sections~\ref{theory} and \ref{conv1}.
\begin{figure*}[t]
	\centering
	\subfloat[]
	{\includegraphics[width=0.3\columnwidth]{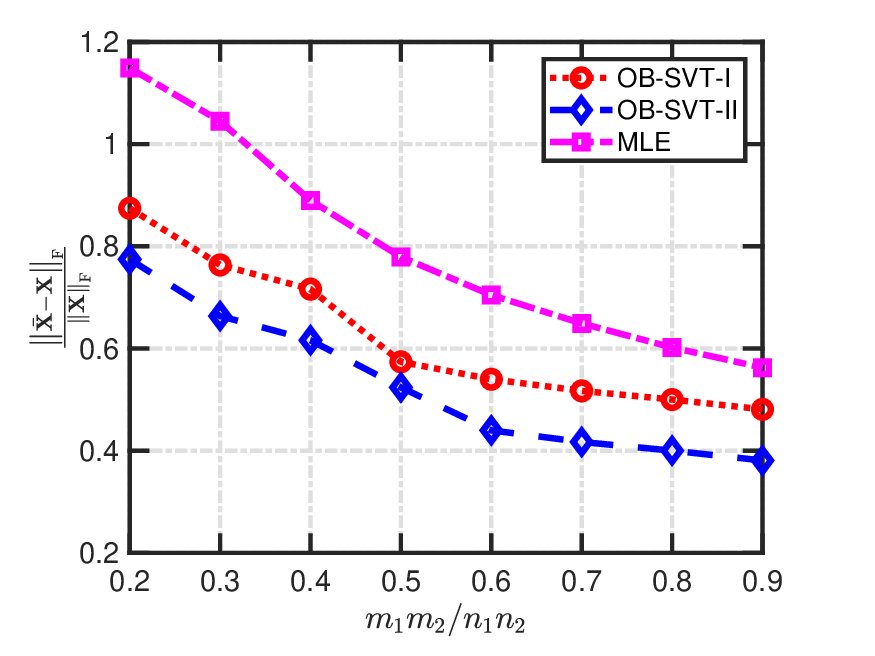}}\qquad
	\subfloat[]
	{\includegraphics[width=0.3\columnwidth]{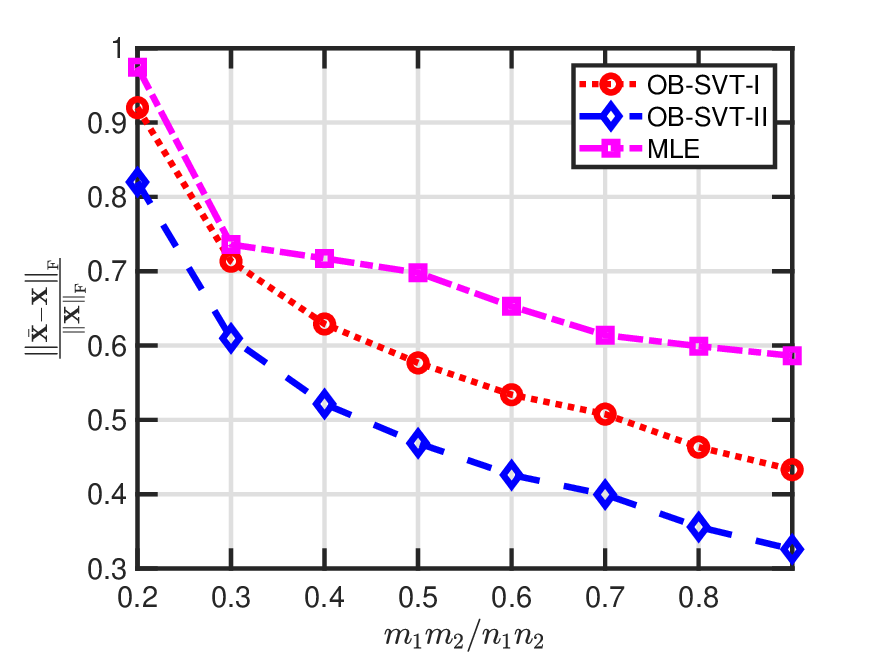}}\qquad
	\subfloat[]
	{\includegraphics[width=0.3\columnwidth]{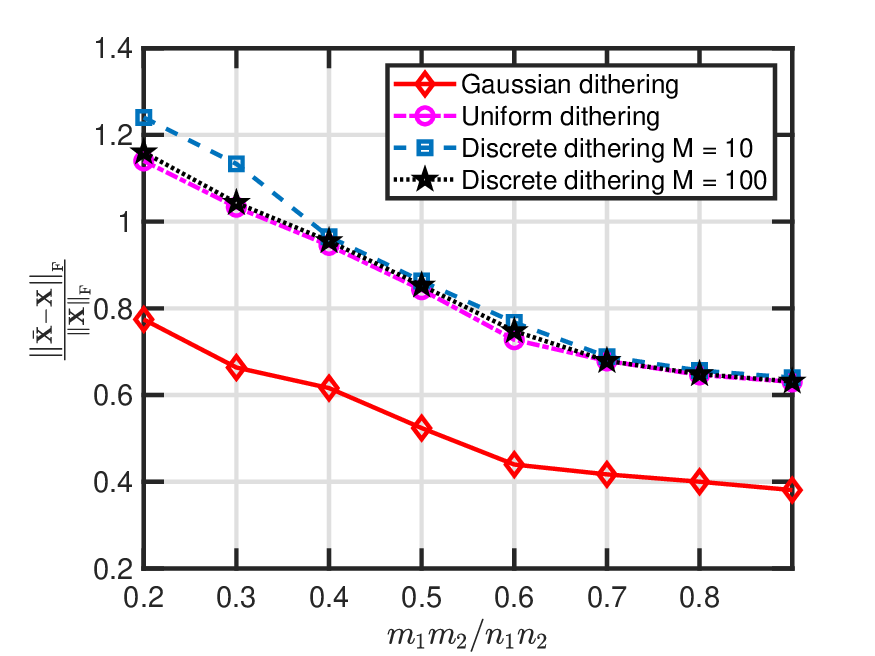}}\qquad\\
	\subfloat[]
	{\includegraphics[width=0.3\columnwidth]{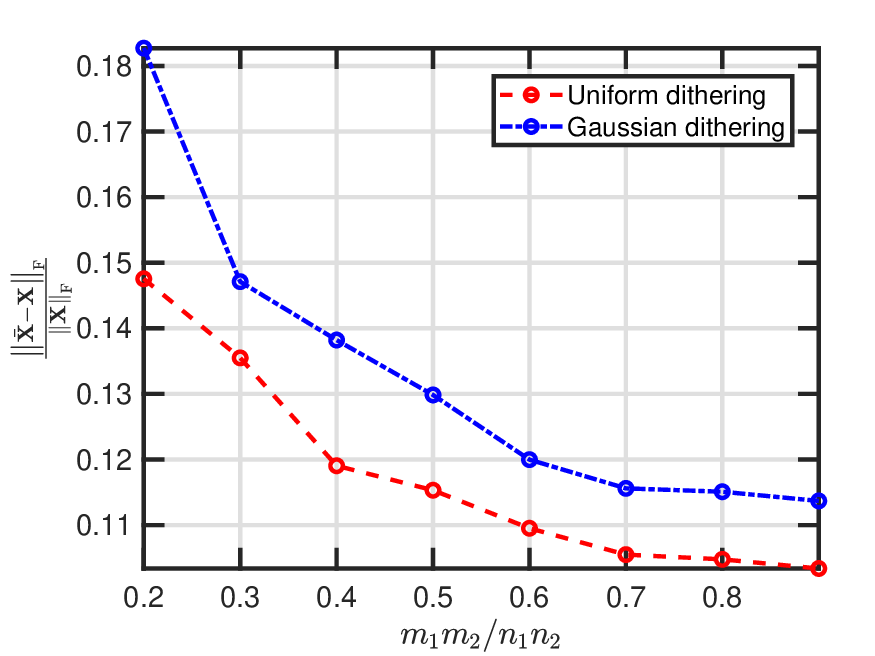}}\qquad
	\subfloat[]
	{\includegraphics[width=0.3\columnwidth]{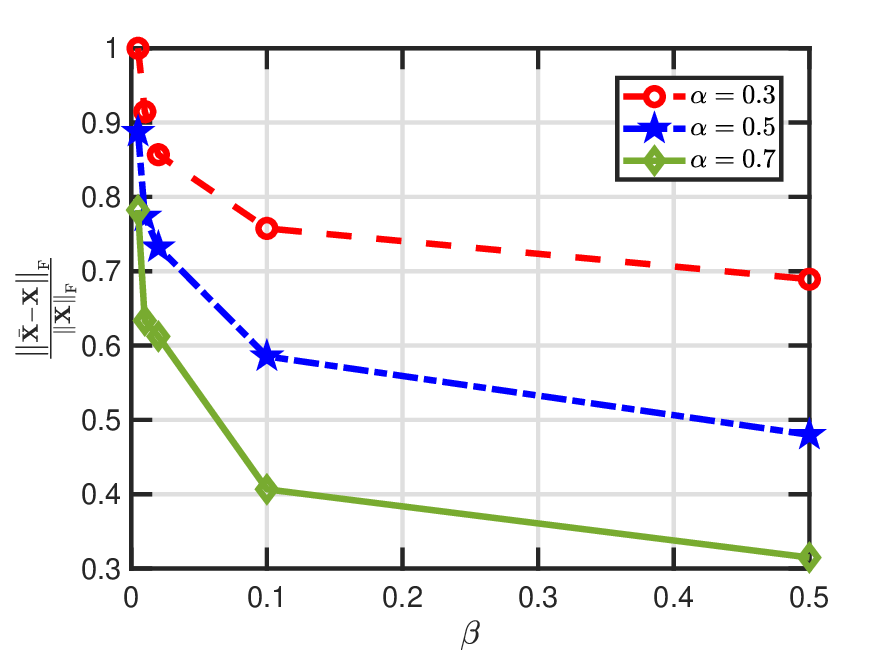}}\qquad
	\subfloat[]
	{\includegraphics[width=0.3\columnwidth]{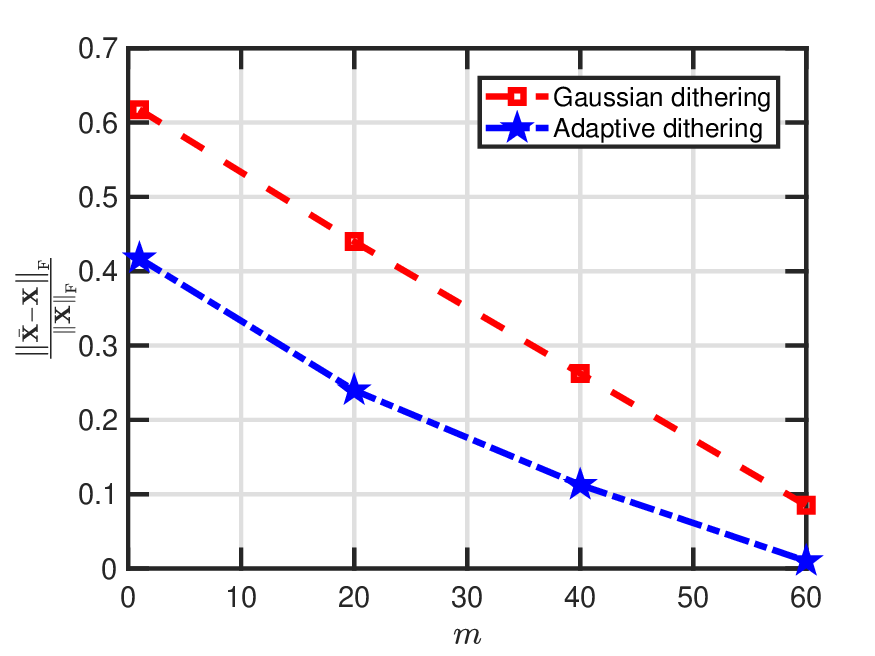}}
	\caption{(a) Comparison between the reconstruction performance of OB-SVT-I, OB-SVT-II and the MLE \cite{davenport20141} in terms of the relative error 
		when a $500\times 500$ matrix $\mbX$ with the rank $r=10$ is corrupted by Gaussian noise. (b) Comparison between the reconstruction performance of OB-SVT-I, OB-SVT-II and the MLE \cite{davenport20141} in terms of the relative error when the measurements are corrupted by Poisson noise. (c) Comparing different dithering schemes, Gaussian, uniform and discrete with quantization levels $M=10$ and $M=100$, for the one-bit MC with Gaussian measurements corrupted by Gaussian noise. (d) Comparison between uniform dithering and Gaussian dithering schemes for the input measurements generated based on uniform distribution. (e) The performance of randomized OB-SVT algorithm for different values of $\beta=s/m^{\prime}$ and $\alpha=m^{\prime}/n_1 n_2$. (f) Improvement of reconstruction accuracy through the OB-SVT-II method as the number of dithering sequences grows large, and comparison between the performance of Gaussian dithering and the proposed adaptive dithering scheme.
	}
	\vspace{-10pt}
	\label{figure_1}
\end{figure*}
In Fig.~\ref{figure_1}b, we have considered the same settings as in Fig.~\ref{figure_1}a with the exception that the noise matrix $\mbZ=[z_{i,j}]$ adheres to the Poisson distribution with the parameter $\lambda=0.5$. Similar to our previous observation, OB-SVT-I and OB-SVT-II show a better reconstruction results in terms of relative error compared to the MLE approach. With the same settings as in Fig.~\ref{figure_1}a, Fig.~\ref{figure_1}c compares the reconstruction performance of OB-SVT-II in terms of relative error, when Gaussian, Uniform, and discrete dithers with $M=10$ and $M=100$ are utilized. We have generated the uniform random dithers as $\left\{\boldsymbol{\uptau}^{(\ell)}\sim\mathcal{U}_{[-a,a]}\right\}_{\ell=1}^{m}$, where $a$ denotes the DR of $\mathcal{P}_{\Omega}\left(\mbY\right)$. The discrete dithers can be generated as explained in Section~\ref{OB}. As can be observed in Fig.~\ref{figure_1}c, in the case of Gaussian input matrix, OB-SVT-II with Gaussian dither shows a better reconstruction performance compared to that of OB-SVT-II with uniform and discrete dithers. 
Also, it can be observed that by increasing the value of $M$ in the discrete dither, the reconstruction performance of OB-SVT-II tends to the case of uniform dithering. The reason behind this is simply followed by the distribution of discrete dither as described in \eqref{Stephania_8}.
In the next experiment, we constructed a random $500\times 500$ matrix $\mbX=[X_{i,j}]$ with rank $r=10$ by forming $\mbX=\mbX_{1}\mbX_{2}^{\top}$, where $\mbX_{1}$ and $\mbX_{2}$ are $500\times 10$ matrices with entries drawn i.i.d. from the uniform distribution $\mathcal{U}_{[0,1]}$. Fig.~\ref{figure_1}d compares the reconstruction performance of OB-SVT-II in terms of relative error, when Gaussian and uniform dithering are utilized.
In this case, the reconstruction performance of OB-SVT-II with uniform dither outperforms the performance of OB-SVT-II with Gaussian dither. 
Note that the results reported in Fig.~\ref{figure_1}d are related to the noiseless scenario. Define the parameters $\alpha=\frac{m^{\prime}}{n_1n_2}$ and $\beta=\frac{s}{m_1m_2}$, where $s$ denotes the sketch length in randomized OB-SVT algorithm. By considering the same settings as in Fig.~\ref{figure_1}a (with no noise), Fig.~\ref{figure_1}e illustrates the reconstruction performance of randomized OB-SVT in terms of relative error respect to $\beta\in\{0.005,0.01,0.02,0.1,0.5\}$ while $\alpha\in\{0.3,0.5,0.7\}$. Based on Theorem~\ref{ST_theo3}, increasing the sketch length $s$ (or increasing the value of $\beta$ here) results in an accelerated convergence rate which is also numerically validated in Fig.~\ref{figure_1}e. With the same settings as in Fig.~\ref{figure_1}a, Fig.~\ref{figure_1}f compares the reconstruction performance
of OB-SVT-II in terms of relative error respect to the number of random dither sequences $m$ when Gaussian and adaptive dithering scheme are utilized.
As shown, increasing the number of random dithers leads to a better reconstruction accuracy. This observation validates the upper recovery bound obtained in Theorem~\ref{thr_1}.

Concerning the computational cost of our algorithms, a significant portion of the complexity arises from the SVD utilized in the algorithms. It is important to note that the MLE problem is addressed using projected gradient descent, wherein the SVT operator is applied to the gradient descent step of the MLE objective at each iteration. As discussed in \cite{davenport20141}, due to the absence of a closed-form operator for solving the MLE problem, the gradient descent algorithm involving the computation of gradients is required at each iteration, in addition to updating the Lagrangian multiplier.
In contrast, the OB-SVT algorithms incur lower computational costs since there is no need to optimize any MLE objective. Instead, we simply update the Lagrangian multiplier at each iteration, involving only a straightforward matrix-vector multiplication. 
We demonstrate the improvement in computational complexity in terms of CPU time in Table~\ref{table_1}. The CPU time for both the OB-SVT-II algorithm and the MLE solver is reported for matrices of rank $5$ with varying sizes. This comparison is performed using a stopping criterion of $\mathrm{NMSE}\leq 0.8$ for both algorithms. The results reveal a substantial reduction in CPU time for our proposed algorithm compared to MLE.
\begin{table}[t]
\caption{CPU time comparison (in seconds) between OB-SVT-II and the MLE methods for rank $5$ input matrices $\mbX\in\mathbb{R}^{n\times n}$.}
\centering
\begin{tabular}{  c || c | c | c   }
\hline
$n$ & $100$ & $300$ & $500$ \\[0.5 ex]
\hline
\textbf{OB-SVT-II} & $0.19$ & $0.30$ & $0.46$ \\
\hline
\textbf{MLE} & $0.70$ & $1.90$ & $4.60$ \\
\hline
\end{tabular}
\label{table_1}
\end{table}
\section{Discussion}
\label{sec:summ}
This study explores the impact of dithered one-bit sensing on matrix completion, framing the challenge as a nuclear norm minimization problem incorporating linear inequality feasibility constraints from one-bit samples. The adaptation of the SVT algorithm to accommodate these constraints, both in scenarios with and without noise, outperformed the regularized MLE method in terms of recovery performance. Incorporating the concept of dithered one-bit MC, we examined noise characteristics within the problem context, relating the distance between entries and dithering sequences to a form of noise. Furthermore, we established a relationship between the number of samples and the suppression of upper bound on average distances between dithers and measurements, ultimately leading to obtain signal reconstruction performance. Numerical discussions extended to discrete dithering and comparisons with uniform and Gaussian scenarios. This paper mainly addressed one-bit MC using nuclear norm minimization. Low-rank matrix factorization could offer an alternative avenue for further research. The convergence rate of algorithms based on low-rank matrix factorization remains an active research area \cite{chi2019nonconvex}. Our study showed promising outcomes with the sketch-and-project method using randomized sketching matrices and partial measurements. Further potential lies in investigating randomized SVD \cite{derezinski2022sharp,oh2017fast}.

\appendices

\section{Proof of Theorem~\ref{thr_1}}
\label{Stephania_1}
We begin the proof by presenting the following lemma:
\begin{lemma}
\label{lem_1}
In the settings of Definition~\ref{def_1}, we have
\begin{equation}
\label{a_6}
\operatorname{Pr}\left(\sup_{\mbX\in\mathcal{K}_r}\left|T_{\mathrm{ave}}(\mbX)-\frac{\alpha}{2}-\frac{\|\mbX\|_{\mathrm{F}}^2}{2\alpha n_1n_2}\right|\geq\epsilon\right)\leq 2e^{-\frac{\epsilon^2 mm^{\prime}}{4\alpha^2}},
\end{equation}
where $\epsilon$ is a positive value.
\end{lemma}
\begin{IEEEproof}
For simplicity of notation, denote $d_{i,j}^{(\ell)}=\left|X_{i,j}-\tau_{i,j}^{(\ell)}\right|$ by $d=\left|X-\tau\right|$. Then, we can write
\begin{equation}
\label{a_29}
\begin{aligned}
\mathbb{E}_{\tau}\left\{d\right\}&=\frac{1}{2\alpha}\int_{-\alpha}^{\alpha}\left|X-\tau\right|\,d\tau\\&=\frac{1}{2\alpha}\left[\int_{-\alpha}^{X}X-\tau\,d\tau+\int_{X}^{\alpha}\tau-X\,d\tau\right]=\frac{\alpha}{2}+\frac{X^2}{2\alpha}.
\end{aligned}
\end{equation}
Therefore, we have 
\begin{equation}
\label{a_30}
\begin{aligned}
\mathbb{E}_{\tau}\left\{T_{\mathrm{ave}}(\mbX)\right\}&=\frac{1}{mm^{\prime}}\sum_{(i,j)\in\Omega}\sum_{\ell=1}^{m}\frac{\alpha}{2}+\frac{X_{i,j}^2}{2\alpha}\\&=\frac{\alpha}{2}+\frac{\|\mathcal{P}_{\Omega}\left(\mbX\right)\|_{\mathrm{F}}^2}{2\alpha m^{\prime}}.
\end{aligned}
\end{equation}
Computing the expected value of \eqref{a_30} respect to the randomness of $\Omega$ leads to 
\begin{equation}
\label{ghombol}
\mathbb{E}_{\tau,\Omega}\left\{T_{\mathrm{ave}}(\mbX)\right\}=\frac{\alpha}{2}+\frac{\|\mbX\|_{\mathrm{F}}^2}{2\alpha n_1n_2}.
\end{equation} 
In the following lemma, we present the Hoeffding's inequality for bounded random variables:
\begin{lemma}\cite[Theroem~2.2.5]{vershynin2018high}
\label{lemma_hoeffding} 
Let $\left\{X_i\right\}^{n}_{i=1}$ be independent, bounded random variables
satisfying $X_i \in [a_i, b_i]$, then for any $t > 0$ it holds that
\begin{equation}
\operatorname{Pr}\left(\left|\frac{1}{n} \sum_{i=1}^n\left(X_i-\mathbb{E}\left\{ X_i\right\}\right)\right| \geq t\right)\leq 2 e^{-\frac{2 n^2 t^2}{\sum_{i=1}^n\left(b_i-a_i\right)^2}}.  
\end{equation}
\end{lemma}
Note that for each random variable $d_{i,j}^{(\ell)}$ we have $0\leq d_{i,j}^{(\ell)}\leq 2\alpha$. Then, following Lemma~\ref{lemma_hoeffding}, we can write
\begin{equation}
\label{a_31}
\operatorname{Pr}\left(\left|T_{\mathrm{ave}}(\mbX)-\frac{\alpha}{2}-\frac{\|\mbX\|_{\mathrm{F}}^2}{2\alpha n_1n_2}\right|\geq\epsilon\right)\leq 2e^{-\frac{\epsilon^2 mm^{\prime}}{2\alpha^2}}.
\end{equation}
As we consider the supremum over all $\mbX\in\mathcal{K}_r$, it is necessary to multiply the resulting probability by the covering number of the defined set. It is straightforward to verify that the covering number of $\rho$-balls required to cover the set $\mathcal{K}_{r}$ is upper bounded by
\begin{equation}
\mathcal{N}\left(\mathcal{K}_{r}, \left\|\cdot\right\|_{\mathrm{F}}, \rho\right)\leq \left(1+\frac{2\alpha\sqrt{n_1n_2}}{\rho}\right)^{(n_1+n_2)r},
\end{equation}
which can be further upper bounded by
\begin{equation}
\label{cov}
\begin{aligned}
\mathcal{N}\left(\mathcal{K}_{r}, \left\|\cdot\right\|_{\mathrm{F}},\rho\right)&\leq e^{(n_1+n_2)r\log\left(1+\frac{2\alpha\sqrt{n_1n_2}}{\rho}\right)}\\&\leq e^{\frac{2\alpha(n_1+n_2)r\sqrt{n_1n_2}}{\rho}}.
\end{aligned} 
\end{equation} 
Based on \eqref{cov}, the Kolmogorov $\rho$-entropy of the set $\mathcal{K}_r$ is upper bounded by
\begin{equation}
\label{cov23}
\begin{aligned}
\mathcal{H}\left(\mathcal{K}_{r},\rho\right)\leq \frac{2\alpha(n_1+n_2)r\sqrt{n_1n_2}}{\rho}.
\end{aligned} 
\end{equation}
To achieve the probability at least $1-2e^{-\frac{\epsilon^2 mm^{\prime}}{4\alpha^2}}$, it is sufficient to write
\begin{equation}
e^{\frac{2\alpha (n_1+n_2)r\sqrt{n_1n_2}}{\rho}}\leq e^{\frac{\epsilon^2 mm^{\prime}}{4\alpha^2}}, 
\end{equation}
or equivalently,  
\begin{equation}
mm^{\prime}\geq \frac{8\alpha^{3}(n_1+n_2)r\sqrt{n_1n_2}}{\epsilon^2\rho}, 
\end{equation}
which proves the lemma.
\end{IEEEproof}
Define $\mbZ=\frac{1}{2}(\mbX+\bar{\mbX})$. We can write $Z_{i,j}-\tau_{i,j}^{(\ell)}=\frac{1}{2}\left(X_{i,j}-\tau_{i,j}^{(\ell)}+\bar{X}_{i,j}-\tau_{i,j}^{(\ell)}\right)$. According to the consistent reconstruction property defined in Definition~\ref{def_2}, we can write
\begin{equation}
\label{a_33}
\left|Z_{i,j}-\tau_{i,j}^{(\ell)}\right|=\frac{1}{2}\left(\left|X_{i,j}-\tau_{i,j}^{(\ell)}\right|+\left|\bar{X}_{i,j}-\tau_{i,j}^{(\ell)}\right|\right).
\end{equation}
Based on \eqref{a_33}, for all $(i,j)\in\Omega$ and $\ell\in[m]$, we can write
\begin{equation}
\label{a_34}
T_{\mathrm{ave}}(\mbZ)=\frac{1}{2}\left[T_{\mathrm{ave}}(\mbX)+T_{\mathrm{ave}}(\bar{\mbX})\right].
\end{equation}
Following Lemma~\ref{lem_1}, with a failure probability at most $2e^{-\frac{\epsilon^2 mm^{\prime}}{4\alpha^2}}$,
we have
\begin{equation}
\label{a_9}
\frac{\|\mbZ\|_{\mathrm{F}}^2}{2\alpha n_1n_2}\geq T_{\mathrm{ave}}(\mbZ)-\frac{\alpha}{2}-\epsilon.
\end{equation}
Combining the results of \eqref{a_34} and \eqref{a_9} leads to
\begin{equation}
\label{a_10}
\begin{aligned}
\frac{\|\mbZ\|_{\mathrm{F}}^2}{2\alpha n_1n_2}&\geq\frac{1}{2}\left[T_{\mathrm{ave}}(\mbX)+T_{\mathrm{ave}}(\bar{\mbX})\right]-\frac{\alpha}{2}-\epsilon\\&\geq\frac{1}{2}\left[\frac{\|\mbX\|_{\mathrm{F}}^2}{2\alpha n_1n_2}+\frac{\alpha}{2}-\epsilon+\frac{\|\bar{\mbX}\|_{\mathrm{F}}^2}{2\alpha n_1n_2}+\frac{\alpha}{2}-\epsilon\right]-\frac{\alpha}{2}-\epsilon\\&=\frac{1}{4\alpha n_1n_2}\left[\|\mbX\|_{\mathrm{F}}^2+\|\bar{\mbX}\|_{\mathrm{F}}^2\right]-2\epsilon.
\end{aligned}
\end{equation}
Based on the definition of $\mbZ$, we can rewrite \eqref{a_10} in terms of $\mbX$ and $\bar{\mbX}$ as follows
\begin{equation}
\label{a_11}
\|\mbX+\bar{\mbX}\|_{\mathrm{F}}^2\geq2\left(\|\mbX\|_{\mathrm{F}}^2+\|\bar{\mbX}\|_{\mathrm{F}}^2\right)-16\alpha n_1n_2\epsilon.
\end{equation}
By the parallelogram law, we conclude that
\begin{equation}
\label{a_12}
\begin{aligned}
\|\mbX-\bar{\mbX}\|_{\mathrm{F}}^2&=2\left[\|\mbX\|_{\mathrm{F}}^2+\|\bar{\mbX}\|_{\mathrm{F}}^2\right]-\|\mbX+\bar{\mbX}\|_{\mathrm{F}}^2\\&\leq16\alpha n_1n_2\epsilon.
\end{aligned}
\end{equation}
Denote $\rho=4\sqrt{\epsilon\alpha n_1n_2}$. Then, according to Lemma~\ref{lem_1}, with a failure probability at most $2e^{-\frac{\epsilon^2 mm^{\prime}}{4\alpha^2}}$, we have $\|\mbX-\bar{\mbX}\|_{\mathrm{F}}\leq\rho$ and 
\begin{equation}
mm^{\prime}\gtrsim\sqrt{r}\max\left(n_1,n_2\right),
\end{equation}
which completes the proof.

\section{Proof of Theorem~\ref{thr_2}}
\label{proof_thr_2}
Define $\mbZ=\frac{1}{2}(\mbX+\bar{\mbX})$. We can write $Z_{i,j}-\tau_{i,j}^{(\ell)}=\frac{1}{2}\left(X_{i,j}-\tau_{i,j}^{(\ell)}+\bar{X}_{i,j}-\tau_{i,j}^{(\ell)}\right)$. Define the set $\mathcal{G}$ such that
\begin{equation}
\label{p_1}
\operatorname{sgn}\left(X_{i,j}-\tau_{i,j}^{(\ell)}\right)\neq\operatorname{sgn}\left(\bar{X}_{i,j}-\tau_{i,j}^{(\ell)}\right),\quad (i,j,\ell)\in\mathcal{G}.
\end{equation}
For all $(i,j,\ell)\in\Omega\times[m]\setminus\mathcal{G}$, we have 
\begin{equation}
\label{p_2}
\left|Z_{i,j}-\tau_{i,j}^{(\ell)}\right|=\frac{1}{2}\left(\left|X_{i,j}-\tau_{i,j}^{(\ell)}\right|+\left|\bar{X}_{i,j}-\tau_{i,j}^{(\ell)}\right|\right),
\end{equation} 
due to the consistent reconstruction property. It can be simply verified that for all $(i,j,\ell)\in\mathcal{G}$, we have 
\begin{equation}
\label{p_3}
\begin{aligned}
\left|Z_{i,j}-\tau_{i,j}^{(\ell)}\right|&=\frac{1}{2}\left(\left|X_{i,j}-\tau_{i,j}^{(\ell)}\right|+\left|\bar{X}_{i,j}-\tau_{i,j}^{(\ell)}\right|\right)\\&-\min\left(\left|X_{i,j}-\tau_{i,j}^{(\ell)}\right|,\left|\bar{X}_{i,j}-\tau_{i,j}^{(\ell)}\right|\right).
\end{aligned}
\end{equation} 
Taking average over all $(i,j)\in\Omega,\ell\in[m]$ leads to
\begin{equation}
\label{p_4}
T_{\mathrm{ave}}(\mbZ)=\frac{1}{2}\left[T_{\mathrm{ave}}(\mbX)+T_{\mathrm{ave}}(\bar{\mbX})\right]-R,
\end{equation} 
where 
\begin{equation}
\label{p_5}
R=\frac{1}{mm^{\prime}}\sum_{(i,j,\ell)\in\mathcal{G}}\min\left(\left|X_{i,j}-\tau_{i,j}^{(\ell)}\right|,\left|\bar{X}_{i,j}-\tau_{i,j}^{(\ell)}\right|\right).
\end{equation} 
Based on Theorem~\ref{thr_1}, if $mm^{\prime}\gtrsim\sqrt{r}\max\left(n_1,n_2\right)$ is met, then with a failure probability at most $2e^{-\frac{\epsilon^2 mm^{\prime}}{4\alpha^2}}$, we have
\begin{equation}
\label{p_6}
\frac{\|\mbZ\|_{\mathrm{F}}^2}{2\alpha n_1n_2}\geq T_{\mathrm{ave}}(\mbZ)-\frac{\alpha}{2}-\epsilon,
\end{equation}
which together with \eqref{p_4} results in
\begin{equation}
\label{p_7}
\begin{aligned}
\frac{\|\mbZ\|_{\mathrm{F}}^2}{2\alpha n_1n_2}&\geq\frac{1}{2}\left[T_{\mathrm{ave}}(\mbX)+T_{\mathrm{ave}}(\bar{\mbX})\right]-R-\frac{\alpha}{2}-\epsilon\\&\geq\frac{1}{2}\left[\frac{\|\mbX\|_{\mathrm{F}}^2}{2\alpha n_1n_2}+\frac{\|\bar{\mbX}\|_{\mathrm{F}}^2}{2\alpha n_1n_2}+\alpha-2\epsilon\right]-R-\frac{\alpha}{2}-\epsilon\\&=\frac{1}{4\alpha n_1n_2}\left[\|\mbX\|_{\mathrm{F}}^2+\|\bar{\mbX}\|_{\mathrm{F}}^2\right]-R-2\epsilon.
\end{aligned}
\end{equation}
Based on the definition of $\mbZ$ and by the parallelogram law, we can write
\begin{equation}
\label{p_8}
\|\mbX-\bar{\mbX}\|_{\mathrm{F}}^2\leq8\alpha n_1n_2R+16\alpha n_1n_2\epsilon.
\end{equation}
According to Theorem~\ref{thr_1}, with a failure probability at most $2e^{-\frac{\epsilon^2 mm^{\prime}}{4\alpha^2}}$, the value of $R$ is bounded as
\begin{equation}
\label{p_9}
R\leq\frac{|\mathcal{G}|}{mm^{\prime}}\alpha,
\end{equation}
which together with \eqref{p_8} completes the proof.
\section{Convergence Analysis of OB-SVT-I}
\label{Stephania_2}
Define $\mbe_k=\boldsymbol{y}^{(k)}-\boldsymbol{y}$, and 
\begin{equation}
\label{kvm_16}
\begin{aligned}
P\left(\mbX^{(k)}\right)&=\boldsymbol{y}^{(k-1)}+\delta_k \left(\mbt-\mathcal{A}\left(\mbX^{(k)}\right)\right),\\P\left(\mbX\right)&=\boldsymbol{y}+\delta_k \left(\mbt-\mathcal{A}\left(\mbX\right)\right).
\end{aligned}
\end{equation} 
Since $\boldsymbol{y}=\left(P\left(\mbX\right)\right)^{+}$ \cite{eaves1971basic}, the error norm is written as
\begin{equation}
\label{MySt_1}
\begin{aligned}
\left\|\mbe_k\right\|^2_2 & =\left\|\left(P\left(\mbX^{(k)}\right)\right)^{+}-\left(P\left(\mbX\right)\right)^{+}\right\|^2_2 \\
& \leq\left\|\boldsymbol{y}^{(k-1)}-\boldsymbol{y}+\delta_k\left(\mathcal{A}\left(\mbX\right)-\mathcal{A}\left(\mbX^{(k)}\right)\right)\right\|^2_2\\
&= \left\|\mbe_{k-1}\right\|^2_2+ 2\delta_k \mbe^{\top}_{k-1} \left(\mathcal{A}\left(\mbX\right)-\mathcal{A}\left(\mbX^{(k)}\right)\right)\\
&\quad\quad\quad\quad\quad+\delta^{2}_{k} \left\|\mathcal{A}\left(\mbX\right)-\mathcal{A}\left(\mbX^{(k)}\right)\right\|^2_2.
\end{aligned}
\end{equation}
As shown in \cite{cai2010singular}, we have
\begin{equation}
\label{ghadboland}
\mbe^{\top}_{k-1} \left(\mathcal{A}\left(\mbX\right)-\mathcal{A}\left(\mbX^{(k)}\right)\right)\leq -\left\|\mbX^{(k)}-\mbX\right\|^2_{\mathrm{F}},
\end{equation}
and 
\begin{equation}
\label{ghadboland1}
\begin{aligned}
\left\|\mathcal{A}\left(\mbX\right)-\mathcal{A}\left(\mbX^{(k)}\right)\right\|^2_{2}&=\left\|\boldsymbol{\mathcal{B}} \operatorname{vec}\left(\mbX\right)-\boldsymbol{\mathcal{B}} \operatorname{vec}\left(\mbX^{(k)}\right)\right\|^2_{2}\\&\leq \sigma^{2}_{\mathrm{max}}\left(\boldsymbol{\mathcal{B}}\right) \left\|\mbX^{(k)}-\mbX\right\|^2_{\mathrm{F}}.
\end{aligned}
\end{equation}
Combining \eqref{ghadboland} and \eqref{ghadboland1} with \eqref{MySt_1} results in
\begin{equation}
\label{MySt_2}
\begin{aligned}
\left\|\mbe_k\right\|^2_2 \leq \left\|\mbe_{k-1}\right\|^2_2- \left(2\delta_k-\delta^2_k \sigma^{2}_{\mathrm{max}}\left(\boldsymbol{\mathcal{B}}\right)\right) \left\|\mbX^{(k)}-\mbX\right\|^2_{\mathrm{F}}.
\end{aligned}
\end{equation}
Herein, we obtain the maximum singular value of one-bit MC matrix $\boldsymbol{\mathcal{B}}$ as follows: 
\begin{equation}
\label{proof_th}
\begin{aligned}
\mbW &=\boldsymbol{\mathcal{B}}^{\top}\boldsymbol{\mathcal{B}}\\&=  \left[\begin{array}{c|c|c}
\mbP^{\top}\bOmega^{(1)} &\cdots &\mbP^{\top}\bOmega^{(m)}
\end{array}\right]\left[\begin{array}{c|c|c}
\mbP^{\top}\bOmega^{(1)} &\cdots &\mbP^{\top}\bOmega^{(m)}
\end{array}\right]^{\top}\\
&=\mbP^{\top}\bOmega^{(1)}\bOmega^{(1)}\mbP+\cdots+\mbP^{\top}\bOmega^{(m)}\bOmega^{(m)}\mbP
= m\mbP^{\top}\mbI\mbP=m\mbP^{\top}\mbP,
\end{aligned}
\end{equation}
which means that the singular values of $\boldsymbol{\mathcal{B}}$ are $\left\{\sigma_{\boldsymbol{\mathcal{B}}}\right\}=\sqrt{m}\left\{\sigma_{i\mbP}\right\}$. Also, the singular values of the permutation matrix $\mbP$ are all equal to one. Therefore, the maximum (or minimum) singular value of one-bit MC matrix $\boldsymbol{\mathcal{B}}$ is given by $\sigma^{2}_{\mathrm{max}}\left(\boldsymbol{\mathcal{B}}\right)=m$,
which proves the theorem.
\section{Convergence Analysis of OB-SVT-II}
\label{Stephania_3}
Define $\mbe_k=\boldsymbol{y}^{(k)}-\boldsymbol{y}$. 
Then, the error norm is written as
\begin{equation}
\label{MySt_10}
\begin{aligned}
\left\|\mbe_k\right\|^2_2 
& =\left\|\boldsymbol{y}^{(k-1)}-\boldsymbol{y}+\delta_k\left(\mbt-\mathcal{A}\left(\mbX^{(k)}\right)\right)^{+}\right\|^2_2\\
&= \left\|\mbe_{k-1}\right\|^2_2+ 2\delta_k \mbe^{\top}_{k-1} \left(\mbt-\mathcal{A}\left(\mbX^{(k)}\right)\right)^{+}\\
&+\delta^{2}_{k} \left\|\left(\mbt-\mathcal{A}\left(\mbX^{(k)}\right)\right)^{+}\right\|^2_2.
\end{aligned}
\end{equation}
Since a nonnegative vector, $\delta_k\left(\mbt-\mathcal{A}\left(\mbX^{(k)}\right)\right)^{+}$, is added to $\boldsymbol{y}^{(k-1)}$ at each iteration, we can write
\begin{equation}
\label{ghadboland❤️}
\begin{aligned}
\boldsymbol{y}^{(k-1)}+\delta_k \left(\mbt-\mathcal{A}\left(\mbX^{(k)}\right)\right)^{+}\preceq \boldsymbol{y},
\end{aligned}
\end{equation}
which together with 
$\delta_k\left(\mbt-\mathcal{A}\left(\mbX^{(k)}\right)\right)^{+}\succeq\mathbf{0}$ leads to 
\begin{equation}
\label{ghadboland❤️❤️}
\delta_k \mbe^{\top}_{k-1} \left(\mbt-\mathcal{A}\left(\mbX^{(k)}\right)\right)^{+}\leq -\delta^2_k \left\|\left(\mbt-\mathcal{A}\left(\mbX^{(k)}\right)\right)^{+}\right\|^2_2.
\end{equation} 
Combining \eqref{ghadboland❤️❤️} with \eqref{MySt_10} leads to
\begin{equation}
\label{kvm_17}
\left\|\mbe_k\right\|^2_2\leq\left\|\mbe_{k-1}\right\|^2_2-\delta^2_k \left\|\left(\mbt-\mathcal{A}\left(\mbX^{(k)}\right)\right)^{+}\right\|^2_2.
\end{equation}
According to the Hoffman inequality \cite{hoffman2003approximate}, \cite[Theorem~4.2]{leventhal2010randomized}, we have
\begin{equation}
\label{ghadboland1❤️}
\begin{aligned}
\sigma^{2}_{\mathrm{min}}\left(\boldsymbol{\mathcal{B}}\right) \left\|\mbX^{(k)}-\mbX\right\|^2_{\mathrm{F}}\leq \left\|\left(\mbt-\mathcal{A}\left(\mbX^{(k)}\right)\right)^{+}\right\|^2_{2}.
\end{aligned}
\end{equation}
Combining \eqref{ghadboland1❤️} with \eqref{kvm_17} results in
\begin{equation}
\label{MySt_20}
\begin{aligned}
\left\|\mbe_k\right\|^2_2 \leq \left\|\mbe_{k-1}\right\|^2_2-\delta^2_k \sigma^{2}_{\mathrm{min}}\left(\boldsymbol{\mathcal{B}}\right) \left\|\mbX^{(k)}-\mbX\right\|^2_{\mathrm{F}},
\end{aligned}
\end{equation}
where according to the proof presented in Appendix~\ref{Stephania_2},
the minimum singular value of one-bit MC matrix $\boldsymbol{\mathcal{B}}$ is given by $\sigma^{2}_{\mathrm{min}}\left(\boldsymbol{\mathcal{B}}\right)=m$ which proves the theorem.
\section{Convergence Analysis of Randomized OB-SVT}
\label{Stephania_4}
Define $\mbe_k=\boldsymbol{y}^{(k)}_s-\boldsymbol{y}_s$, $\widehat{t}=\mbS\mbt$, and $\widehat{\boldsymbol{\mathcal{B}}}=\mbS\boldsymbol{\mathcal{B}}$. Then, the error norm is written as
\begin{equation}
\label{MySt_100}
\begin{aligned}
\left\|\mbe_k\right\|^2_2 
& =\left\|\boldsymbol{y}^{(k-1)}_s-\boldsymbol{y}_s+\delta_k\left(\widehat{\mbt}-\widehat{\boldsymbol{\mathcal{B}}}\operatorname{vec}\left(\mbX^{(k)}\right)\right)^{+}\right\|^2_2\\
&= \left\|\mbe_{k-1}\right\|^2_2+ 2\delta_k \mbe^{\top}_{k-1} \left(\widehat{\mbt}-\widehat{\boldsymbol{\mathcal{B}}}\operatorname{vec}\left(\mbX^{(k)}\right)\right)^{+}\\
&+ \delta^{2}_{k} \left\|\left(\widehat{\mbt}-\widehat{\boldsymbol{\mathcal{B}}}\operatorname{vec}\left(\mbX^{(k)}\right)\right)^{+}\right\|^2_2.
\end{aligned}
\end{equation}
Since a nonnegative vector, $\delta_k\left(\widehat{\mbt}-\widehat{\boldsymbol{\mathcal{B}}}\operatorname{vec}\left(\mbX^{(k)}\right)\right)^{+}$, is added to $\boldsymbol{y}^{(k-1)}_s$ at each iteration, we can write
\begin{equation}
\label{myghadboland❤️}
\begin{aligned}
\boldsymbol{y}^{(k-1)}_s+\delta_k\left(\widehat{\mbt}-\widehat{\boldsymbol{\mathcal{B}}}\operatorname{vec}\left(\mbX^{(k)}\right)\right)^{+}\preceq\boldsymbol{y}_s,    
\end{aligned}
\end{equation}
which together with 
$\delta_k\left(\widehat{\mbt}-\widehat{\boldsymbol{\mathcal{B}}}\operatorname{vec}\left(\mbX^{(k)}\right)\right)^{+}\succeq\mathbf{0}$ leads to
\begin{equation}
\label{myghadboland❤️❤️}
\begin{aligned}
\delta_k \mbe^{\top}_{k-1} \left(\widehat{\mbt}-\widehat{\boldsymbol{\mathcal{B}}}\operatorname{vec}\left(\mbX^{(k)}\right)\right)^{+}\leq-\delta^2_k \left\|\left(\widehat{\mbt}-\widehat{\boldsymbol{\mathcal{B}}}\operatorname{vec}\left(\mbX^{(k)}\right)\right)^{+}\right\|^2_2.
\end{aligned}
\end{equation}
Combining \eqref{myghadboland❤️❤️} with \eqref{MySt_100} leads to
\begin{equation}
\label{MySt_2000}
\begin{aligned}
\left\|\mbe_k\right\|^2_2 \leq \left\|\mbe_{k-1}\right\|^2_2-\delta^2_k \left\|\left(\widehat{\mbt}-\widehat{\boldsymbol{\mathcal{B}}}\operatorname{vec}\left(\mbX^{(k)}\right)\right)^{+}\right\|^2_2.
\end{aligned}
\end{equation}
Since the infinite norm is constrained by the norm-$2$, \eqref{MySt_2000} is rewritten as
\begin{equation}
\label{MySt_200000}
\begin{aligned}
\left\|\mbe_k\right\|^2_2 \leq \left\|\mbe_{k-1}\right\|^2_2-\delta^2_k \left\|\left(\widehat{\mbt}-\widehat{\boldsymbol{\mathcal{B}}}\operatorname{vec}\left(\mbX^{(k)}\right)\right)^{+}\right\|^2_{\infty}.
\end{aligned}
\end{equation}
By taking the expectation over the error, we have
\begin{equation}
\label{MySt_20000}
\begin{aligned}
\mathbb{E}_{\mbS}\left\{\left\|\mbe_k\right\|^2_2\right\} \leq \left\|\mbe_{k-1}\right\|^2_2-\delta^2_k \mathbb{E}_{\mbS}\left\{\left\|\left(\widehat{\mbt}-\widehat{\boldsymbol{\mathcal{B}}}\operatorname{vec}\left(\mbX^{(k)}\right)\right)^{+}\right\|^2_{\infty}\right\}.
\end{aligned}
\end{equation}
We can apply the Jensen's inequality as follows:
\begin{equation}
\label{kvm_18}
\begin{aligned}
\mathbb{E}_{\mbS}\{\|(\widehat{\mbt}-\widehat{\boldsymbol{\mathcal{B}}}&\operatorname{vec}(\mbX^{(k)}))^{+}\|^2_{\infty}\}\geq\\&(\mathbb{E}_{\mbS}\{\|(\widehat{\mbt}-\widehat{\boldsymbol{\mathcal{B}}}\operatorname{vec}(\mbX^{(k)}))^{+}\|_{\infty}\})^2.
\end{aligned}
\end{equation}
By taking advantage from the estimate for the maximum of independent normal random variables \cite[Section~2.5.2]{vershynin2018high},
the following bound is obtained:
\begin{equation}
\label{St_1400}
\begin{aligned}
&\mathbb{E}_{\mbS}\left\{\left\|\left(\widehat{\mbt}-\widehat{\boldsymbol{\mathcal{B}}}\operatorname{vec}\left(\mbX^{(k)}\right)\right)^{+}\right\|_{\infty}\right\} 
\\
&\geq c\left\|\left(\Tilde{\mbt}-\Tilde{\boldsymbol{\mathcal{B}}}\operatorname{vec}\left(\mbX^{(k)}\right)\right)^{+}\right\|_2 \sqrt{\log s},
\end{aligned}
\end{equation}
where $c$ is a positive value, $\mbs_{t}$ and $\mbg_{t}$ 
correspond to the $t$-th rows of $\mbS$ and $\mbG$, 
while $\Tilde{\mbt}\in\mathbb{R}^{m^{\prime}}$ and $\Tilde{\boldsymbol{\mathcal{B}}}\in\mathbb{R}^{m^{\prime}\times n_1n_2}$ denote the blocks of $\mbt$ and $\boldsymbol{\mathcal{B}}$, respectively.
Combining \eqref{kvm_18} and \eqref{St_1400} with \eqref{MySt_20000} results in
\begin{equation}
\label{kvm_19}
\begin{aligned}
\mathbb{E}_{\mbS}\left\{\left\|\mbe_k\right\|^2_2\right\} \leq \left\|\mbe_{k-1}\right\|^2_2-\delta^2_k c^2 \log s \left\|\left(\Tilde{\mbt}-\Tilde{\boldsymbol{\mathcal{B}}}\operatorname{vec}\left(\mbX^{(k)}\right)\right)^{+}\right\|_2^2.
\end{aligned}
\end{equation}
According to the Hoffman inequality \cite{hoffman2003approximate}, \cite[Theorem~4.2]{leventhal2010randomized},
we have
\begin{equation}
\label{myghadboland1❤️}
\begin{aligned}
\sigma^{2}_{\mathrm{min}}\left(\Tilde{\boldsymbol{\mathcal{B}}}\right) \left\|\mbX^{(k)}-\mbX\right\|^2_{\mathrm{F}}\leq \left\|\left(\Tilde{\mbt}-\Tilde{\boldsymbol{\mathcal{B}}}\operatorname{vec}\left(\mbX^{(k)}\right)\right)^{+}\right\|^2_{2}.
\end{aligned}
\end{equation}
Based on this result, \eqref{kvm_19} can be rewritten as
\begin{equation}
\label{MySt_200}
\begin{aligned}
\mathbb{E}\left\{\left\|\mbe_k\right\|^2_2\right\} \leq \left\|\mbe_{k-1}\right\|^2_2-\delta^2_k c^2\sigma^{2}_{\mathrm{min}} \left(\Tilde{\boldsymbol{\mathcal{B}}}\right) \log s\left\|\mbX^{(k)}-\mbX\right\|^2_{\mathrm{F}}.
\end{aligned}
\end{equation}
Similar to the proof provided in Appendix~\ref{Stephania_2}, one can simply show that $\sigma^{2}_{\mathrm{min}} (\Tilde{\boldsymbol{\mathcal{B}}})=1$ which proves the theorem.
\section{Proof of Lemma~\ref{St_lemma}}
\label{Stephania_5}
According to the first step of OB-SVT algorithms, we write
\begin{equation}
\label{ST-AR1}
\left\|\mbX^{(k)}-\mbX\right\|^2_{\mathrm{F}}= \left\|\mathcal{D}_\tau\left(\mathcal{A}^{\star}\left(\boldsymbol{y}^{(k-1)}\right)\right)-\mathcal{D}_\tau\left(\mathcal{A}^{\star}\left(\boldsymbol{y}\right)\right)\right\|^2_{\mathrm{F}}.
\end{equation}
Consider an operator function $\mathcal{G}_f$ applied to a matrix $\mbX$ with rank $r$ as follows:
$\mathcal{G}_f\left(\mbX\right) = \sum^{r}_{k=1} f\left(\sigma_k\right)\mbu_k\mbv^{\mathrm{H}}_k$,
where $\left\{\sigma_k, \mbu_k, \mbv_k\right\}^{r}_{k=1}$ are singular values of $\mbX$ and its corresponding singular vectors, and $f$ is a $L$-Lipschitz continuous projector function. As comprehensively discussed in \cite[Theorem~4.2]{andersson2016operator}, the following relation holds for two matrices $\mbX_1$ and $\mbX_2$ belonging to the Hilbert space $\mathcal{H}$:
\begin{equation}
\label{Lip1}
\left\|\mathcal{G}_f\left(\mbX_1\right)-\mathcal{G}_f\left(\mbX_2\right)\right\|_{\mathrm{F}} \leq L \left\|\mbX_1-\mbX_2\right\|_{\mathrm{F}}.
\end{equation}
In the OB-SVT algorithm, the function $f$ 
is an operator which compares the singular values of $\mbX$ with a fixed threshold $\tau$ and then 
eliminate the smaller ones, which is straightforward that it satisfies \eqref{Lip1} with $L=1$. Therefore, one can conclude
\begin{equation}
\label{Lip2}
\left\|\mathcal{D}_\tau\left(\mbX_1\right)-\mathcal{D}_\tau\left(\mbX_2\right)\right\|^2_{\mathrm{F}}\leq \left\|\mbX_1-\mbX_2\right\|^2_{\mathrm{F}},~\forall\mbX_1, \mbX_2 \in \mathcal{H}.
\end{equation}
Combining this result with \eqref{ST-AR1} results in
\begin{equation}
\label{kvm_20}
\left\|\mbX^{(k)}-\mbX\right\|^2_{\mathrm{F}}\leq\left\|\mathcal{A}^{\star}\left(\boldsymbol{y}^{(k-1)}\right)-\mathcal{A}^{\star}\left(\boldsymbol{y}\right)\right\|^2_{\mathrm{F}}.
\end{equation} 
Also, we have 
\begin{equation}
\label{ST-AR2}
\left\|\mathcal{A}^{\star}\left(\boldsymbol{y}^{(k-1)}\right)-\mathcal{A}^{\star}\left(\boldsymbol{y}\right)\right\|^2_{\mathrm{F}}\leq \sigma^2_{\mathrm{max}}\left(\mathcal{A}^{\star}\right) \left\|\boldsymbol{y}^{(k-1)}-\boldsymbol{y}\right\|^2_2,
\end{equation}
which leads to
\begin{equation}
\label{ST-AR20}
\left\|\mbX^{(k)}-\mbX\right\|^2_{\mathrm{F}} \leq \sigma^2_{\mathrm{max}}\left(\mathcal{A}^{\star}\right) \left\|\boldsymbol{y}^{(k-1)}-\boldsymbol{y}\right\|^2_{2}.
\end{equation}
It is straightforward to verify that $\sigma^2_{\mathrm{max}}\left(\mathcal{A}^{\star}\right)=\sigma^2_{\mathrm{max}}\left(\boldsymbol{\mathcal{B}}^{\dagger}\right)$. Note that 
\begin{equation}
\label{Stephanie_Arian}
\begin{aligned}
\boldsymbol{\mathcal{B}}^{\dagger} = \left(\mbP^{\top} \Tilde{\bOmega}^{\top}\Tilde{\bOmega}\mbP\right)^{-} \boldsymbol{\mathcal{B}}^{\top}= \left(m\mbP^{\top}\mbP\right)^{-} \boldsymbol{\mathcal{B}}^{\top}
= \frac{1}{m} \boldsymbol{\mathcal{B}}^{\top},
\end{aligned}
\end{equation}
where $\Tilde{\bOmega}$ is defined in \eqref{eq:9}. Therefore, the maximum singular value of $\boldsymbol{\mathcal{B}}^{\dagger}$ is given by
\begin{equation}
\label{StephanieEamaz}
\begin{aligned}
\sigma^2_{\mathrm{max}}\left(\boldsymbol{\mathcal{B}}^{\dagger}\right) 
= \frac{1}{m^2} \sigma^2_{\mathrm{max}}\left(\boldsymbol{\mathcal{B}}^{\top}\right)=
\frac{1}{m^2} \sigma^2_{\mathrm{max}}\left(\boldsymbol{\mathcal{B}}\right)= \frac{1}{m}, 
\end{aligned}
\end{equation}
which proves the lemma. 

\section{Proof of Proposition~\ref{prop_1}}
\label{ghadboland_1}
Using the same notations as introduced in Appendix~\ref{Stephania_2}, the error norm is written as
\begin{equation}
\label{MySteph_1}
\begin{aligned}
\left\|\mbe_k\right\|^2_2
& \leq\left\|\boldsymbol{y}^{(k-1)}-\boldsymbol{y}+\delta_k\left(\mathcal{A}\left(\mbX+\mbZ\right)-\mathcal{A}\left(\mbX^{(k)}\right)\right)\right\|^2_2\\
&\leq \left\|\mbe_{k-1}\right\|^2_2+ 2\delta_k \mbe^{\top}_{k-1} \left(\mathcal{A}\left(\mbX+\mbZ\right)-\mathcal{A}\left(\mbX^{(k)}\right)\right)\\
&+ \delta^{2}_{k} \left\|\mathcal{A}\left(\mbX+\mbZ\right)-\mathcal{A}\left(\mbX^{(k)}\right)\right\|^2_2.
\end{aligned}
\end{equation}
According to Theorem~\ref{ST_theo1}, since $\sigma^2_{\textrm{max}}\left(\boldsymbol{\mathcal{B}}\right)=m$, we have
\begin{equation}
\label{MySteph_2}
\begin{aligned}
\left\|\mbe_k\right\|^2_2 \leq \left\|\mbe_{k-1}\right\|^2_2- \left(2\delta_k-m\delta^2_k\right) \left\|\mbX^{(k)}-\mbX-\mbZ\right\|^2_{\mathrm{F}},
\end{aligned}
\end{equation}
where based on the triangle inequality, we can write
\begin{equation}
\label{MySteph_3}
\begin{aligned}
\left\|\mbe_k\right\|^2_2 \leq \left\|\mbe_{k-1}\right\|^2_2- \left(2\delta_k-m\delta^2_k\right) \left(\frac{1}{2}\left\|\mbX^{(k)}-\mbX\right\|^2_{\mathrm{F}}-\left\|\mbZ\right\|^2_{\mathrm{F}}\right).
\end{aligned}
\end{equation}
To guarantee the convergence of OB-SVT-I in the presence of noise, the error norm must be a nonincreasing function with respect to the iteration. Therefore, the noise level must satisfy $\left\|\mbZ\right\|^2_{\mathrm{F}}\leq \frac{1}{2}\left\|\mbX^{(k)}-\mbX\right\|^2_{\mathrm{F}}$.
For the OB-SVT-II algorithm, we use the same notations as introduced in Appendix~\ref{Stephania_3}, the error norm is written as
\begin{equation}
\label{kvm_21}
\begin{aligned}
\left\|\mbe_k\right\|^2_2 &= \left\|\mbe_{k-1}\right\|^2_2+ 2\delta_k \mbe^{\top}_{k-1} \left(\mbt-\mathcal{A}\left(\mbX^{(k)}\right)\right)^{+}\\
&+\delta^{2}_{k} \left\|\left(\mbt-\mathcal{A}\left(\mbX^{(k)}\right)\right)^{+}\right\|^2_2,
\end{aligned}
\end{equation}
or simply
\begin{equation}
\label{kvm_22}
\left\|\mbe_k\right\|^2_2\leq\left\|\mbe_{k-1}\right\|^2_2-\delta^2_k \left\|\left(\mbt-\mathcal{A}\left(\mbX^{(k)}\right)\right)^{+}\right\|^2_2.
\end{equation}
According to the Hoffman inequality \cite{hoffman2003approximate}, \cite[Theorem~4.2]{leventhal2010randomized}, we have
\begin{equation}
\label{kvm_23}
\begin{aligned}
\sigma^{2}_{\mathrm{min}}\left(\boldsymbol{\mathcal{B}}\right) \left\|\mbX^{(k)}-\mbX-\mbZ\right\|^2_{\mathrm{F}}\leq \left\|\left(\mbt-\mathcal{A}\left(\mbX^{(k)}\right)\right)^{+}\right\|^2_{2}.
\end{aligned}
\end{equation}
Combining \eqref{kvm_23} with \eqref{kvm_22} leads to
\begin{equation}
\label{kvm_24}
\left\|\mbe_k\right\|^2_2\leq\left\|\mbe_{k-1}\right\|^2_2-m\delta^2_k\left\|\mbX^{(k)}-\mbX-\mbZ\right\|^2_{\mathrm{F}}.
\end{equation} 
Since $\frac{1}{2}\left\|\mbX^{(k)}-\mbX\right\|^2_{\mathrm{F}}-\left\|\mbZ\right\|^2_{\mathrm{F}}\leq\left\|\mbX^{(k)}-\mbX-\mbZ\right\|^2_{\mathrm{F}}$, we can rewrite \eqref{kvm_24} as
\begin{equation}
\label{kvm_25}
\left\|\mbe_k\right\|^2_2\leq\left\|\mbe_{k-1}\right\|^2_2-\frac{1}{2}m\delta^2_k\left\|\mbX^{(k)}-\mbX\right\|^2_{\mathrm{F}}+m\delta_k^2\left\|\mbZ\right\|^2_{\mathrm{F}}.
\end{equation} 
To guarantee the convergence of OB-SVT-II in the presence of noise, the error norm must be a nonincreasing function with respect to the iteration. Therefore, the noise level must satisfy
$\left\|\mbZ\right\|^2_{\mathrm{F}}\leq \frac{1}{2}\left\|\mbX^{(k)}-\mbX\right\|^2_{\mathrm{F}}$,
which proves the proposition. Verifying the proof of the similar proposition for the randomized OB-SVT is as straightforward as the proofs of other propositions. To establish the convergence of the randomized OB-SVT algorithm in the presence of noise, it suffices to demonstrate that the \emph{expectation} of the error is non-increasing. This can be straightforwardly proven in a similar manner to the preceding proofs.
\bibliographystyle{IEEEtran}
\bibliography{references}

\end{document}